\documentclass[twoside,11pt, letter]{article}

\usepackage{fixltx2e} 

\usepackage{cmap}

\usepackage[T1]{fontenc}
\usepackage[utf8]{inputenc}

\usepackage{verbatim}

\usepackage{booktabs}       

\usepackage{nicefrac}       

\usepackage{microtype}      

\usepackage{graphicx} 

\usepackage{color}
\usepackage[colorlinks]{hyperref}
\usepackage{url}

\usepackage{geometry} 
\usepackage[toc,page]{appendix}
\usepackage{mathrsfs}

\usepackage{amsmath, amsbsy, amsgen, amscd, amsthm, amsfonts, amssymb} 
\usepackage[all]{xy}

\usepackage{framed}

\newtheorem{theorem}{Theorem}

\theoremstyle{myRem}

\theoremstyle{myDef}

\newtheorem{theorem*}[theorem]{Theorem}   
\newtheorem{lemma*}{Lemma} 
\newtheorem{corollary*}{Corollary} 
\newtheorem{remark*}{Remark}
\newtheorem{example*}{Example}
\newtheorem{definition*}{Definition}
\newtheorem{proposition*}{Proposition}

\def\cB{{\cal B}}
\def\cC{{\cal C}}

\def\cF{{\cal F}}

\def\cK{{\cal K}}

\newcommand{\real}{\ensuremath{\mathbb{R}}}

\newcommand{\E}{\ensuremath{\mathbb{E}}}

\newcommand{\mprob}{\ensuremath{\mathbb{P}}}
\newcommand{\Gtilde}{\ensuremath{\widetilde{G}}}

\newcommand{\ERalpha}{\ensuremath{\overline{\regret}_{T}^{\alpha}}}
\newcommand{\ER}{\ensuremath{\overline{\regret}_{T}}}
\newcommand{\lhat}{\ensuremath{\widehat{\ell}}}

\providecommand{\mathbbm}{\mathbb} 

\newcommand{\rr}{\mathbbm{R}}

\newcommand{\expect}{\mathbb{E}}
\newcommand{\prob}{\mathbb{P}}

\newcommand{\ind}{\pmb{1}}

\newcommand{\graphset}{G}
\newcommand{\vertexset}{V}
\newcommand{\edgeset}{E}
\newcommand{\sourceset}{\mathcal{S}}
\newcommand{\advset}{\mathcal{A}}

\newcommand{\regret}{R}

\DeclareMathOperator*{\argmax}{{arg \, max}}
\DeclareMathOperator*{\argmin}{{arg \, min}}

\title{Adversarial Influence Maximization}

\author{Justin Khim\thanks{Department of Statistics, University of Pennsylvania, Philadelphia, PA 19104.}
\and 
Varun Jog\thanks{Department of Electrical \& Computer Engineering,
       University of Wisconsin,
       Madison, WI 53706.}
\and 
Po-Ling Loh\footnotemark[2] \thanks{Department of Statistics,       University of Wisconsin,
       Madison, WI 53706.}}
\date{January 19, 2019}
\begin{document}

\maketitle

\begin{abstract}
We consider the problem of influence maximization in fixed networks for contagion models in an adversarial setting. The goal is to select an optimal set of nodes to seed the influence process, such that the number of influenced nodes at the conclusion of the campaign is as large as possible. We formulate the problem as a repeated game between a player and adversary, where the adversary specifies the edges along which the contagion may spread, and the player chooses sets of nodes to influence in an online fashion. We establish upper and lower bounds on the minimax pseudo-regret in both undirected and directed networks.
\end{abstract}

\section{Introduction}

Many data sets in contemporary scientific applications possess some underlying network structure~\cite{New03}. Popular examples include data collected from social media websites such as Facebook and Twitter~\cite{AdaEyt03, LibKle07}, or electrocortical recordings gathered from a network of firing neurons~\cite{Spo11}. 
An important application of network science arises in marketing, where
researchers have studied the importance of word-of-mouth advertising for decades~\cite{katz1966}.
More recently, methods have been proposed by marketing researchers to quantify the importance of word-of-mouth marketing in online social networks in both theory and practice \cite{watts2007, trusov2009}.
Subsequent empirical studies suggest that word-of-mouth marketing has a significant effect in online social networks \cite{aral2011, seiler2017}.
At the same time, computer scientists have analyzed the problem of viral marketing from an optimization-theoretic perspective~\cite{DomRic01, LesEtal07, chen2013book}, where the goal is to select an optimal set of influencers to encourage product adoption in an online social network.
This has led to rigorous theoretical guarantees that hold for stochastic models of word-of-mouth advertising inspired by physics and epidemiology, and the scope of the spread is quantified using a notion known as influence \cite{KemEtal03}. In social networks, edges represent potential interactions between individuals, and the problem of influence maximization corresponds to identifying subsets of individuals on which to impress an idea so that information spreads as widely as possible subject to an advertising budget.

Formally, the influence of a subset of nodes is defined as the expected number of influenced individuals in a network at the conclusion of a spread, starting from an initial configuration where only the specified nodes are influenced. Even when the influence function is assumed to be computable for any subset using a black-box method in unit time, it is not clear whether influence maximization may be performed (exactly or approximately) in polynomial time, since searching over all subsets of $k$ nodes is exponential in the number of nodes. Accordingly, the body of work in theoretical computer science has mostly focused on specific spreading models that give rise to nice properties such as submodularity, implying that a greedy algorithm for influence maximization leads to a constant-factor approximation of the optimal set~\cite{KemEtal03, KemEtal05, BorEtal14}. Other related work includes predicting when knowledge becomes viral; limiting the spread of information through carefully positioned interventions~\cite{CheEtal14, drakopoulos2016}; or competitive settings of influence maximization, e.g.\ competing for votes or market share \cite{bharathi2007, he2013, grabisch2017}. 


A significant shortcoming in the analysis of stochastic spreading models is the fact that the parameters characterizing the spread of influence are generally assumed to be known, allowing for approximate evaluation of the influence function (either by analytic methods or simulation). However, such an assumption is not always practical. In the case of independent cascade models or linear threshold models, where parameters correspond to edge weights in the network, one might even question a scientist's prior knowledge of the precise network structure.
To address these issues, some authors have studied the interesting question of accurately learning the influence function itself in a stochastic spreading model based on observing multiple rounds of infection \cite{NarEtal15, LeiEtal15, he2016}. Another approach involves a notion of ``robust influence maximization,'' where the parameters are only specified to lie in fixed confidence sets, and the goal is to obtain a set of source vertices that approximately maximizes the true influence function, possibly in a worst-case sense~\cite{CheEtal16, HeKem16}. Robust influence maximization methods may also be model-dependent, meaning that a robust algorithm designed for the independent cascade model may lead to a severely non-optimal solution if the influence spread actually follows linear threshold model. Indeed, the parameters describing different models, as well as the nature of uncertainties permitted in them, may be completely different.
Further, it is unclear that popular models of influence are good apprixmations of real-world behavior \cite{goel2012, hu2017}.

In this paper, we take a rather different approach toward the problem of unknown spreading parameters that also avoids assumptions about a particular spreading mechanism. As discussed in more detail in Section~\ref{SecModels}, we only assume knowledge of an underlying fixed graph representing the paths along which a influence  may spread, where the case of no prior knowledge corresponds to a complete graph. We formulate the influence maximization problem as an online game, where a ``player'' must make sequential decisions about the next seed set to choose based on observing the behavior of the spread in previous ``rounds'' of the game. Here, a round represents a particular instance of an influence process initialized from the specific seed nodes from beginning to end. We allow an ``adversary'' to choose the path of influence on each round in a completely arbitrary manner, as long as the process may only spread along edges of the graph---in particular, this setting subsumes the stochastic models usually adopted in the influence maximization literature, while allowing for much more general spreading mechanisms (e.g., information does not necessarily propagate in an i.i.d.\ manner over all rounds of the game). Note that the adversary's strategy may be so arbitrary as to be ``unlearnable.'' Thus, instead of simply trying to maximize the aggregate number of influenced vertices across all rounds, we seek to develop player strategies that bound the ``regret'' of the player, defined as the difference between the total number of vertices influenced using the player's strategy and the number of vertices that would have been influenced if the player had adopted the best constant choice of source set in hindsight. Such notions are taken from the literature on multi-armed bandits and online learning theory~\cite{auer2002, plg2006}, and adapted to the present setting.

Our main contribution is to derive upper and lower bounds on the pseudo-regret for various adversarial and player strategies. We study both directed and undirected networks, where in the latter setting, contagion is allowed to spread in both directions when an edge is chosen by the adversary. 
Furthermore, we derive lower bounds for the minimax pseudo-regret when the underlying network is a complete graph, where the supremum is taken over all adversarial strategies and the infimum is taken over all player strategies. Our upper and lower bounds match up to constant factors in the case of directed networks. Notably, the bounds also agree with the usual rate for pseudo-regret in multi-armed bandits, showing that no new information is gained by the player by exploiting network structure. On the other hand, a gap exists between our upper and lower bounds for undirected networks, leaving open the possibility that the player may leverage the additional information from the network to incur less regret. 
Additionally, the constant factor on the upper bound may be slightly improved, providing further evidence that graph structure may be exploited.
Finally, we demonstrate how to extend our upper bounds to the setting where the player is allowed to choose multiple source vertices on each round. The proposed multi-source player strategy augments the source set sequentially using the single-source strategies as a subroutine, and is based on a general online greedy algorithm proposed by Streeter and Golovin \cite{streeter2007}.

The remainder of our paper is organized as follows: In Section~\ref{SecModels}, we provide some important background on online learning theory and formally define the adversarial spreading model and notions of regret to be studied in our paper.  In Section~\ref{SecOnlineResults}, we present upper and lower bounds for pseudo-regret in the adversarial setting.  We conclude the paper with a selection of open research questions in Section~\ref{SecDiscussion}.
All proofs, as well as a more technical discussion of related work, is contained in the appendices.

\paragraph{\textbf{Notation.}}  
For a set $A$, let $2^A$ denote the power set of $A$. 
When we want to specify that we are taking the expectation with respect to a particular distribution \(p\) of some random variable \(X\), we write \(\expect_{X \sim p}\).
In particular, we often write $\E_{\sourceset \sim p}$ to mean the expectation taken over the player's actions for a fixed set of adversarial actions, which is the same as the conditional expectation with respect to the adversary's actions. Similarly, we write $\E_\advset$ to indicate the conditional expectation with respect to a fixed set of player actions.
\section{Background and preliminaries}
\label{SecModels}

We begin by formally defining the repeated game between the player and adversary and the types of strategies we will analyze in our paper. Next, we introduce the notions of regret we will study, and then connect our setting to related work in the learning theory literature.

\subsection{Adversarial repeated games}

Consider a fixed graph \(\graphset = (\vertexset, \edgeset)\) on \(n\) vertices, which may be directed or undirected. The adversarial influence maximization problem may be described as follows: Repeatedly over $T$ rounds, the player selects an influence seed set \(\sourceset_{t} \subseteq V\), with \(|\sourceset_{t}| = k\), for \(t = 1, \ldots, T\). At the same time, the adversary designates a subset of edges \(\advset_{t} \subseteq \edgeset\) to be ``open.'' A node is considered to be influenced at time $t$ if and only if it is an element of $\sourceset_t$ or is reachable from $\sourceset_t$ via a path of open edges. Note that in the context of influence spreading, the open edges correspond to ties over which influence propagates in that round---importantly, influence only has an opportunity to be transmitted between individuals that interact in the network, but may not necessarily spread over a particular connection on a specific round. In the case when $G$ is an undirected graph, designating an edge to be open allows an influence campaign to spread in both directions. Furthermore, in the directed case, edges may exist in both directions between a given pair of nodes, in which case the adversary may designate both, one, or neither of the edges to be open. For an open edge set $A \subseteq \edgeset$ and influence seed set $S \subseteq \vertexset$, we define $f(A,S)$ to be the fraction of vertices in the graph lying in the influenced set.

To connect our model to the canonical setting of influence maximization, note that \cite{KemEtal05} proposed a very general class of influence models called triggering models, which include the independent cascade and the linear threshold models as special cases. At the beginning of the influence campaign, each node chooses a random ``triggering'' subset of neighbors according to a particular rule, and the incoming edges from those neighbors are designated to be ``active.'' A vertex becomes influenced during the course of the process if and only if a path of active edges exists connecting that vertex to a vertex in the seed set.
Thus, triggering models correspond to a special case of our framework, in which the edge sets are chosen in an i.i.d.\ manner from round to round, and the probability distribution over the edges is determined by the probability rule through which edges are assigned to be active (e.g., according to the linear threshold or independent cascade models).

Next, we describe the classes of strategies $\advset = \{\advset_{t}\}$ and $\sourceset = \{\sourceset_t\}$ available to the adversary and player.
We assume that the adversary is \emph{oblivious} of the player's actions; i.e., at time $t=0$, the adversary must decide on the (possibly random) strategy $\advset$. 
We use \(\mathscr{A}\) to denote the set of oblivious adversary strategies and $\mathscr{A}_d$ to denote the set of deterministic adversary strategies. Turning to the classes of player strategies, we allow the player to choose his or her action at time $t$ based on the feedback provided in response to the joint actions made by the player and adversary on preceding time steps. Although the player knows the edge set $E$ of the underlying graph, we assume that the player only observes the status of edges $(i,j)$ such that either $i$ or $j$ is in the reach of $\sourceset_t$ (in the undirected case), and the player observes the status of every edge $(i,j)$ such that $i$ is in the reach of $\sourceset_t$ (in the directed case). In other words, whereas the player cannot observe the subset of all edges that \emph{would have} propagated influence in the network, he or she will know which edges transmitted influence if reached by the influence cascade initialized using his or her seed set.

Formally, we write \(\mathscr{I}(\advset_{t}, \sourceset_{t})\) to denote the set of edges with status known to the player (i.e., all edges in the subgraph induced by $\advset_t$ belonging to connected components containing nodes in $\sourceset_t$), and we denote $\mathscr{I}^t = (\mathscr{I}(\advset_1, \sourceset_1), \dots, \mathscr{I}(\advset_t, \sourceset_t))$. 
If \(\advset_{t}\) is chosen via a stochastic model such as the independent cascade model with discrete time steps for influence campaign \(t\), our setup technically allows the player knowledge of the status of an edge between two vertices \(u\) and \(v\) if both were actually influenced by some other vertex \(w\).
Realistically we would not want the status of edge \((u, v)\) to be returned as feedback, 
and we could enforce this by positing a model of how each influence campaign proceeds.
However, this distinction does not affect our results or algorithms, and so we do not further restrict the feedback \(\mathscr{I}(\advset_{t}, \sourceset_{t})\).

The player can only make decisions based on the feedback observed in previous rounds, so any allowable player strategy $\{\sourceset_t\}$ has the property that $\sourceset_t$ is a function of $\mathscr{I}^{t - 1}$ (possibly with additional randomization). We denote the class of all player strategies by \(\mathscr{P}\), and denote the subclass of all deterministic player strategies by \(\mathscr{P}_{d}\), meaning that $\sourceset_t$ is a deterministic function of $\mathscr{I}^{t-1}$. Note that strategies $\sourceset_t \in \mathscr{P}_d$ may still be random, due to possible randomization of the adversary, but \emph{conditioned} on $\mathscr{I}^{t-1}$, the choice of $\sourceset_t$ is deterministic.


\subsection{Minimax regret}

The player wishes to devise a strategy that maximizes the aggregate number of influenced nodes up to time $T$. Using the notation from the previous section, we define the \emph{regret} of the player to be
\begin{equation}
\regret_T(\advset, \sourceset)
= 
\sum_{t = 1}^{T}
f(\advset_{t}, \sourceset_{*}) 
- \sum_{t = 1}^{T} f(\advset_{t}, \sourceset_{t}),
\label{eqn: regret}
\end{equation}
where
\[
\sourceset_{*}
=
\argmax_{S: |S| = k} 
\sum_{t = 1}^{T}
f(\advset_{t}, S)
\]
is the optimal fixed set that the player would have chosen in hindsight with full knowledge of the adversary's strategy.

Note that the regret $R_T(\advset, \sourceset)$ may be a random quantity due to randomness in both the adversary's or player's strategies. Accordingly, we will seek to control the \emph{pseudo-regret}
\begin{equation}
\label{EqnPseuReg}
\ER(\advset, \sourceset) := \max_{S: |S| = k} \left\{\E_{\advset, \sourceset} \left[\sum_{t = 1}^{T} f(\advset_{t}, S)  - \sum_{t = 1}^{T} f(\advset_{t}, \sourceset_{t})\right]\right\},
\end{equation}
where the expectation in equation~\eqref{EqnPseuReg} is taken with respect to potential randomization in both $\advset$ and $\sourceset$. As in the standard learning theory literature~\cite{bubeck2012}, recall that the expected regret and pseudo-regret are generally related via the inequality
\begin{equation*}
\ER(\advset, \sourceset) \le \E[R_T(\advset, \sourceset)],
\end{equation*}
although if $\advset \in \mathscr{A}_d$, we have $\ER(\advset, \sourceset) = \E[R_T(\advset, \sourceset)]$.
Our interest in the pseudo-regret rather than the expected regret is purely motivated by the fact that the former quantity is often easier to bound than the latter and that this simplification is common in the literature on bandits.

Finally, we introduce the \emph{scaled regret}
\begin{equation}
\regret^{\alpha}(\advset, \sourceset)
= \alpha\sum_{t = 1}^{T}
 f(\advset_{t}, \sourceset_{*}) 
- \sum_{t = 1}^{T} f(\advset_{t}, \sourceset_{t}),
\label{eqn: regret scaled}
\end{equation}
and the analogous quantity
\begin{equation*}
\ERalpha(\advset, \sourceset) = \max_{S: |S| = k} \left\{\E_{\advset, \sourceset} \left[\alpha \sum_{t = 1}^{T} f(\advset_{t}, S)  - \sum_{t = 1}^{T} f(\advset_{t}, \sourceset_{t})\right]\right\}.
\end{equation*}
Note that $\alpha = 1$ corresponds  to the unscaled version. Our interest in the expression~\eqref{eqn: regret scaled} is again for theoretical purposes, since we may obtain convenient upper bounds on the scaled pseudo-regret in the case $\alpha = 1 - \frac{1}{e}$ using an online greedy algorithm. Note that when $k > 1$, the benchmark greedy algorithms used for influence maximization in the stochastic spreading setting are also only guaranteed to achieve a $\left(1-\frac{1}{e}\right)$-approximation of the truth, so in some sense, the scaled regret~\eqref{eqn: regret scaled} only requires the player to perform comparably well in relation to the appropriately scaled optimal strategy.

\section{Main results}
\label{SecOnlineResults}

In this section, we provide upper and lower bounds for the pseudo-regret. Specifically, we focus on the quantity
\begin{equation*}
\inf_{\sourceset \in \mathscr{P}} \sup_{\advset \in \mathscr{A}}  \ERalpha(\advset, \sourceset),
\end{equation*}
where the supremum is taken over the class of adversarial strategies, and the infimum is taken over the class of player strategies based on the feedback model we have described. In other words, we wish to characterize the hardness of the influence maximization problem in terms of the player's best possible strategy measured with respect to the worst-case game.

A rough outline of our approach is as follows: We establish upper bounds by presenting particular strategies for the player that ensure an appropriately bounded regret under all adversarial strategies. For lower bounds, the general technique is to provide an ensemble of possible actions for the adversary that are difficult for the player to distinguish in the influence maximization problem, which forces the player to incur a certain level of regret.

\subsection{Undirected graphs}

We begin by deriving regret upper bounds for undirected graphs. We initially restrict our attention to the case $k=1$. The proposed player strategy for $k > 1$, and corresponding regret bounds, builds upon the results in the single-source setting.

\subsubsection{Upper bounds for a single source}

Consider a randomized player strategy that selects $\sourceset_t = \{i\}$ with probability $p_{i,t}$. 
The paper \cite{bubeck2012} suggests a method based on the Online Stochastic Mirror Descent (OSMD) algorithm, which is specified by loss estimates $\{\ell_{i,t}\}$ and learning rates $\{\eta_t\}$, as well as a Legendre function $F$. 
Here, we comment on the losses, and in order to avoid excessive technicalities, we defer additional details of the OSMD algorithm to the appendix. 

The most basic loss estimate, which follows from standard bandit theory and ignores all information about the graph, is
\begin{equation}
\label{EqnLnode}
\widehat{\ell}_{i,t}^{\text{node}} = \frac{\ell_{i,t}}{p_{i,t}} \ind _{\sourceset_{t} = \{i\}},
\end{equation}
where $\ell_{i,t} = 1 - f(\advset_t, \{i\})$ is the loss incurred if the player were to choose $\sourceset_t = \{i\}$. Importantly, $\lhat^{\text{node}}_{i,t}$ is always computable for any choice the player makes at time $t$ and is an unbiased estimate of $\ell_{i,t}$.

On the other hand, if $\sourceset_t = \{i\}$ and another node $j$ is influenced (i.e., in the connected component formed by the open edges of $\advset_t$), the player also knows the loss that would have been incurred if $\sourceset_t = \{j\}$, since $f(\advset_t, \{i\}) = f(\advset_t, \{j\})$. This motivates an alternative loss estimate that is nonzero even when $\sourceset_{t} \neq \{i\}$. In particular, we may express
\begin{equation*}
\label{EqnLpair}
\ell_{i,t} = \frac{1}{n} \sum_{j \neq i} \ell_{i,j}^t,
\end{equation*}
where $\ell_{i,j}^t$ is the indicator that $i$ and $j$ are in different connected components formed by the open edges of $\advset_t$. We then define
\begin{equation*}
\widehat{\ell}_{i,t}^{\text{sym}} = \frac{1}{n} \sum_{j \neq i} \ell_{i,j}^t \frac{Z_{ij}}{p_{i, t} + p_{j, t}},
\end{equation*}
where $Z_{ij} = \ind_{\sourceset_t \cap \{i,j\} \neq \emptyset}$. 
Note that  $\lhat_{i,t}^{\text{sym}}$ is also an unbiased estimate for $\ell_{i,t}$. 
The estimator $\lhat_{i,t}^{\text{sym}}$ is always computable by the player, since the value of $\ell_{i,j}^t$ is known by the player whenever $\sourceset_t$ is known. We call $\lhat_{i,t}^{\text{sym}}$ the \emph{symmetric loss}.
Now, we state the following regret bounds:

\begin{theorem} [Symmetric loss, OSMD]
Suppose the player uses the strategy $\sourceset^{\text{sym}}_{\text{OSMD}}$ corresponding to OSMD with the symmetric loss \(\hat{\ell}^{\text{sym}}\) and appropriate parameters. Then the pseudo-regret satisfies the bound
\[
\sup_{\advset \in \mathscr{A}} \overline{\regret}_{T}(\advset, \sourceset^{\text{sym}}_{\text{OSMD}}) \leq  2^{\frac{1}{4}} \sqrt{Tn}.
\]
\label{theorem: symmetric loss osmd}
\end{theorem}

\begin{remark*}
It is instructive to compare the result of Theorem~\ref{theorem: symmetric loss osmd} with analogous regret bounds for generic multi-armed bandits. When the OSMD algorithm is run with the loss estimates~\eqref{EqnLnode}, standard analysis establishes an upper bound of $2^{\frac{3}{2}} \sqrt{Tn}$. Thus, using the symmetric loss, which leverages the graphical nature of the problem, produces slight gains.
\end{remark*}

%
\subsubsection{Lower bounds}

We now establish lower bounds for the pseudo-regret in the case $k = 1$. This furnishes a better understanding of the hardness of the adversarial influence maximization problem. The general approach for deriving lower bounds is to produce a strategy for the adversary that forces the player to incur a certain level of regret regardless of which strategy is chosen.

The intrinsic difficulty of online influence maximization may vary widely depending on the topology of the underlying graph, and methods for deriving lower bounds may also differ accordingly. In the case of a complete graph, we have the following result:

\begin{theorem}
Suppose \(\graphset = \mathcal{K}_n\) is the complete graph on \(n \ge 3\) vertices. Then the pseudo-regret satisfies the lower bound
\[
\frac{2}{243} \sqrt{T}
\leq 
\inf_{\sourceset \in \mathscr{P}} \sup_{\advset \in \mathscr{A}} \ER(\advset, \sourceset).
\]
\label{theorem: complete graph lower}
\end{theorem}

\begin{remark*}
Clearly, a gap exists between the lower bound derived in Theorem~\ref{theorem: complete graph lower} and the upper bound appearing in Theorem~\ref{theorem: symmetric loss osmd}. It is unclear which bound, if any, provides the proper minimax rate. However, note that if the lower bound were tight, it would imply that the proportion of vertices that the player misses by picking suboptimal source sets is constant, meaning the number of additional vertices the optimal source vertex influences is linear in the size of the graph. This differs substantially from the pseudo-regret of order \(\sqrt{n}\) known to be minimax optimal for the standard multi-armed bandit problem (and arises, for instance, in the case of directed graphs, as discussed in the next section).
\end{remark*}

%
\subsubsection{Upper bounds for multiple sources}
\label{subsection: multiple sources}

We now turn to the case $k > 1$, where the player chooses multiple source vertices at each time step. As discussed in Section~\ref{SecModels}, we are interested in bounding the scaled pseudo-regret $\ERalpha(\advset, \sourceset)$ with $\alpha = 1 - \frac{1}{e}$, since it is difficult to maximize the influence even in an offline setting, and the greedy algorithm is only guaranteed to provide a $\left(1-\frac{1}{e}\right)$-approximation of the truth.

Our proposed player strategy is based on an online greedy adaptation of the strategy used in the single-source setting, and the full details are given in the appendix. 
We then have the following result concerning the scaled pseudo-regret:

\begin{theorem}[Symmetric loss, multiple sources]
Suppose $k > 1$ and the player uses the strategy $\sourceset^{\text{sym}, k}_{\text{OSMD}}$ corresponding to the  Online Greedy Algorithm with single-source strategy $\sourceset^{\text{sym}}_{\text{OSMD}}$. Then the scaled pseudo-regret satisfies the bound
\[
\sup_{\advset \in \mathscr{A}} \overline{\regret}_{T}^{(1 - 1 /e)}(\advset, \sourceset^{\text{sym}, k}_{\text{OSMD}})
\leq 
2^{\frac{1}{4}} k 
\sqrt{Tn}.
\]
\label{theorem: symmetric loss multisource}
\end{theorem}

Comparing Theorem~\ref{theorem: symmetric loss multisource} to Theorem~\ref{theorem: symmetric loss osmd}, we see an additional factor of $k$ in the pseudo-regret upper bound. Similar results may be derived when alternative single-source strategies are used as subroutines in the Online Greedy Algorithm.
%
\subsection{Directed graphs}

We now derive upper and lower bounds for the pseudo-regret in the case of directed graphs, when $k = 1$.

\subsubsection{Upper bounds}

The symmetric loss does not have a clear analog in the case of directed graphs.
However, we may still use the node loss estimate for multi-armed bandit problems, given by equation~\eqref{EqnLnode}. This leads to the following upper bound:
\begin{theorem}
Suppose the player uses the strategy $\sourceset^{\text{node}}_{\text{OSMD}}$ corresponding to OSMD with the node loss $\lhat^{\text{node}}$ and appropriate parameters. Then the pseudo-regret satisfies the bound
\[
\sup_{\advset \in \mathscr{A}} \overline{\regret}_{T}(\advset, \sourceset^{\text{node}}_{\text{OSMD}}) \leq  2^{\frac{3}{2}} \sqrt{Tn}.
\]
\label{theorem: node loss osmd}
\end{theorem}

\begin{remark*}
In the case $k > 1$, we may again use the Online Greedy Algorithm used in Section~\ref{subsection: multiple sources} to obtain a player strategy composed of parallel runs of a single-source strategy. If the player uses the single-source strategy $\sourceset^{\text{node}}_{\text{OSMD}}$, we may obtain the scaled pseudo-regret bound
\begin{equation*}
\sup_{\advset \in \mathscr{A}} \overline{\regret}_{T}^{(1 - 1 /e)}(\advset, \sourceset^{\text{node}, k}_{\text{OSMD}})
\leq 
2^{\frac{3}{2}} k 
\sqrt{Tn}.
\end{equation*}
\end{remark*}

%
\subsubsection{Lower bounds}

Finally, we provide a lower bound for the directed complete graph on \(n\) vertices. (This refers to the case where all edges are present and bidirectional.) We have the following result:
\begin{theorem}
Suppose \(\graphset\) is the directed complete graph on \(n\) vertices. Then the pseudo-regret satisfies the lower bound
\[ 
\frac{1}{48 \sqrt{6}} \sqrt{T n}
\leq 
\inf_{\sourceset \in \mathscr{P}}
\sup_{A \in \mathscr{A}} 
\ER(\advset, \sourceset).
\] 
\label{theorem: directed lower bound}
\end{theorem}

Notably, the lower bound in Theorem~\ref{theorem: directed lower bound} matches the upper bound in Theorem~\ref{theorem: node loss osmd}, up to constant factors. Thus, the minimax pseudo-regret for the influence maximization problem is $\Theta(\sqrt{T n})$ in the case of directed graphs. 
In the case of undirected graphs, however (cf.\ Theorem~\ref{theorem: complete graph lower}), we  only obtained a pseudo-regret lower bound of $\Omega(\sqrt{T})$. This is due to the fact that in undirected graphs, one may learn about the loss of other nodes at time $t$ besides the loss at $\sourceset_t$. In contrast, it is possible to construct adversarial strategies for directed graphs that do not provide information regarding the loss incurred by choosing a source vertex other than $\sourceset_t$.

Finally, we remark that a different choice of $\graphset$ might affect the lower bound, since influence maximization is easier for some graph topologies than others. However, Theorem~\ref{theorem: directed lower bound} shows that the case of the complete graph is always guaranteed to incur a pseudo-regret that matches the general upper bound in Theorem~\ref{theorem: node loss osmd}, implying that this is the minimax optimal rate for any class of graphs containing the complete graph.

\section{Discussion}
\label{SecDiscussion}

We have proposed and analyzed player strategies that control the pseudo-regret uniformly across all possible oblivious adversarial strategies. For the problem of single-source influence maximization in complete networks, we have also derived minimax lower bounds that establish the fundamental hardness of the online influence maximization problem. In particular, our lower and upper bounds match up to constant factors in the case of directed complete graphs, implying that our proposed player strategy is in some sense optimal. 

Our work inspires a number of interesting questions for future study. An important open question concerns closing the gap between upper and lower bounds on the minimax pseudo-regret in the case of undirected graphs, to determine whether the feedback available in the influence maximization setting actually makes the online game easier than a standard bandit setting. Furthermore, our lower bounds only hold in the case of complete graphs and single-source influence maximization, and it would be worthwhile to obtain lower bounds that hold for other network topologies and seed sets containing multiple nodes. Our results only address a small subset of problems that may be posed and answered concerning a bandit theory of adversarial influence maximization with edge-level feedback.



\bibliography{refs}
\bibliographystyle{abbrv}

\newpage
\newpage
\begin{appendix}
\section{Related work}
\label{appRelatedWork}

Here, we comment more thoroughly on important relationships between our problem setting and various online games existing in the learning theory literature. A key difference between the graph contagion setting and the standard multi-armed bandit setting is that in the latter case, the only information available to the player on each round is the reward obtained as a consequence of his or her actions. On the other hand, slightly more information is available to the player in our setting, since the player may often deduce additional information about \emph{which vertices would have been} influenced for a different choice of source vertices, based on observing the scope of the influence process for a particular choice of source vertices. As a concrete example, the player knows that exactly the same set of nodes would have been influenced if he or she had chosen to influence a different seed node in the same connected component of the subgraph induced by the influenced nodes and adversarially chosen edges.

Online games with partial monitoring~\cite{BarEtal14} or graph-based feedback~\cite{alon2017} generalize the bandit setting to repeated games in which the player may observe feedback corresponding to various subsets of other actions in addition to or instead of observing the feedback corresponding to his or her actions. Although such games resemble our problem setting, the possible actions available of the player in our case correspond to subsets of nodes of size $k$, leading to a rather complicated feedback graph that is additionally affected by the adversary's actions. Another online game with a similar flavor is the combinatorial prediction setting~\cite{audibert2013}, where the player is allowed to pull a subset of arms on each round, and observes a loss equal to the sum of losses of the pulled arms in the case of bandit feedback or a subvector of losses corresponding to the pulled arms, in the case of semi-bandit feedback \cite{NeuBar13}. Our problem may be cast as a type of combinatorial prediction game with a feedback graph that varies from round to round and is unknown to the player. Note that the combinatorial game with edge semi-bandit feedback has been studied recently in the influence maximization literature~\cite{CheEtal16JMLR, VasEtal15, wen2016, wang2017, olkhovskaya2018}, but these results only apply to stochastic adversaries, rather than the more general non-stochastic framework we study in this paper. Edge semi-bandit feedback refers to the fact that in a directed graph, the player receives feedback about the transmission status of different subsets of edges, corresponding to the outgoing edges from the nodes he or she chooses to seed on each round.
\section{Proofs}
\label{SecProofs}

We now outline the proofs of our main results.

\subsection{Upper bounds for adversarial models}
\label{SubsecOnlineLowerBounds}

In this section, we prove our upper bounds.
To this end, we describe the OSMD algorithm, which generates a sequence of probability distributions $\{p_t\}$ to be employed by the player on successive rounds. Let $\Delta^n \subseteq \real^n$ denote the probability simplex.

\begin{framed}
\noindent \textbf{Online Stochastic Mirror Descent (OSMD) with loss estimates $\{\lhat_{i,t}\}$} \\

\noindent
Given: A Legendre function \(F\) defined on $\real^n$, with associated Bregman divergence
\begin{equation*}
D_F(p,q) = F(p) - F(q) - (p-q)^T \nabla F(q),
\end{equation*}
and a learning rate $\eta > 0$. \\
Output: A stochastic player strategy $\{\sourceset_t\}$. \\

\noindent
Let \(p_1 \in \argmin_{p \in \Delta^n} F(p)\).

\noindent
For each round \(t = 1, \ldots, T\):

\begin{enumerate}
\item[(1)] Draw a vertex \(\sourceset_{t}\) from the distribution $p_t$.

\item[(2)] Compute the vector of loss estimates $\lhat_t = \{\lhat_{i,t}\}$.

\item[(3)] Set \(w_{t + 1} = \nabla F^{*}\left(\nabla F (p_t) - \eta \lhat_t\right)\), where $F^*$ is the convex conjugate of $F$.

\item[(4)] Compute the new distribution \(p_{t+1} = \argmin_{p \in \Delta^n} D_{F}(p, w_{t + 1})\).

\end{enumerate}
\end{framed}

In general, the OSMD algorithm is defined with respect to a compact, convex set $\cK \subseteq \real^n$. The updates are characterized by noisy estimates of the gradient of the loss function, which we may conveniently define to be $\lhat_t$ in the present scenario. For more details and generalizations, we refer the reader to \cite{bubeck2012}.
We will use the following result:

\begin{proposition*}[Theorem 5.10 of \cite{bubeck2012}]
Let the loss functions $\{\ell_{i,t}\}$ be nonnegative and bounded by 1. The strategy $\sourceset$ corresponding to the OSMD algorithm with loss estimates $\lhat$, learning rate $\eta > 0$, and Legendre function \(F_{\psi}\), where $\psi$ is a 0-potential, satisfies the pseudo-regret bound
\[
\sup_{\advset \in \mathscr{A}} \overline{R}_{T}(\advset, \sourceset)
\leq 
\frac{\sup_{p \in \Delta^{n}} F_{\psi}(p) - F_{\psi}(p_1)}{\eta}
+
\frac{\eta}{2}
\sum_{t = 1}^{T} \sum_{i = 1}^{n}
\expect 
\left[
\frac{\lhat_{i,t}^{2}}{(\psi^{-1})' (p_{i,t})}
\right].
\]
\label{proposition: generic loss osmd}
\end{proposition*}

We formally define 0-potentials and the associated Legendre functions in Appendix~\ref{appAdditionalOnlineUpperBoundProofs}. In our analysis, we take \(\psi(x) = \frac{1}{x^2}\), yielding the Legendre function \(F_{\psi}(x) = -2 \sum_{i = 1}^{n} x_{i}^{1/2}\). The pseudo-regret bound in Proposition~\ref{proposition: generic loss osmd} may then be analyzed and bounded accordingly in various settings of interest. Details for the proof of Theorem~\ref{theorem: symmetric loss osmd} are also provided in Appendix~\ref{appAdditionalOnlineUpperBoundProofs}.

%
\subsection{Lower bounds for adversarial models.}
\label{SubSecOnlineLowerBounds}

We now turn to establishing the lower bounds. The proofs of Theorems~\ref{theorem: complete graph lower} and \ref{theorem: directed lower bound} are based on the same general strategy, which is summarized in the following proposition. To unify our results with standard bandit notation (\cite{bubeck2012}), we use the shorthand
\begin{equation*}
X_{i,t} = f(\advset_t, \{i\})
\end{equation*}
to denote the reward incurred at time $t$ when the player chooses $\sourceset_t = \{i\}$. Then
\begin{equation*}
\ER(\advset, \sourceset) = \max_{1 \le i \le n} \E_{\advset, \sourceset} \sum_{t=1}^T (X_{i,t} - X_{\sourceset_t, t}).
\end{equation*}

\begin{proposition*}
Consider a deterministic player strategy $\sourceset \in \mathscr{P}_d$. Let \(\advset^0, \advset^1, \dots, \advset^n\) be stochastic adversarial strategies such that for each $\advset^i$, the set of edges played at time $t$ is independent of the past actions of the adversary. Let \(\prob_{0}, \prob_{1}, \ldots, \prob_{n}\) denote the corresponding measures on the feedback \(\mathscr{I}^{T}\), allowing for possible randomization only in the strategy of the adversary. Let $\E_i$ denote the expectation with respect to $\prob_i$.
Suppose
\begin{equation}
\label{EqnWrongSource}
r 
\leq 
\min_{j \neq i} \expect_i \left[X_{i, t} - X_{j, t}\right], \qquad \forall 1 \le t \le T,
\end{equation}
and 
\begin{equation}
\label{EqnKLbound}
\sum_{i = 1}^{n} KL \left(\prob_0, \prob_i \right) \leq D.
\end{equation}
Then
\begin{equation}
\label{EqnRegInter}
rT \left(\frac{n - 1}{n} - \sqrt{\frac{D}{2n}}\right)
\leq
\frac{1}{n} \sum_{i=1}^n \E_i \sum_{t=1}^T (X_{i,t} - X_{\sourceset_t, t}).
\end{equation}
In particular, if the bounds~\eqref{EqnWrongSource} and~\eqref{EqnKLbound} hold uniformly for all choices of $\sourceset \in \mathscr{P}_d$, then
\begin{equation}
\label{EqnRegFinal}
rT \left(\frac{n - 1}{n} - \sqrt{\frac{D}{2n}}\right) \le \inf_{\sourceset \in \mathscr{P}} \sup_{\advset \in \mathscr{A}} \ER(\advset, \sourceset).
\end{equation}
\label{PropLB}
\end{proposition*}
\begin{remark*}
We remark briefly about the roles of the strategies $\{\advset^i\}$ appearing in Proposition~\ref{PropLB}. In practice, the strategies are chosen to be similar, except selecting $i$ as the source node is slightly more advantageous when the adversary uses strategy $\advset^i$. The strategy $\advset^0$ is a baseline strategy that treats all nodes identically. Thus, the lower bound provided by Proposition~\ref{PropLB} is the product of the cost of an incorrect choice of the source vertex, given by \(r\), and a factor that determines how easy it is to distinguish the adversary strategies from each other, which depends on \(D\).
\end{remark*}

\begin{proof}{Proof of Proposition~\ref{PropLB}.}
We follow the method used in the proof of Theorem 3.5 in \cite{bubeck2012}. We first show how to obtain the bound~\eqref{EqnRegFinal} from the set of uniform bounds~\eqref{EqnRegInter}. Note that for any $\sourceset \in \mathscr{P}$, we have
\begin{align*}
\sup_{\advset \in \mathscr{A}} \ER(\advset, \sourceset) & = \sup_{\advset \in \mathscr{A}} \max_{1 \le i \le n} \E_{\advset, \sourceset} \sum_{t=1}^T (X_{i,t} - X_{\sourceset_t, t}) \\
& \ge \max_{1 \le j \le n} \max_{1 \le i \le n} \E_\sourceset \E_{\advset^j} \sum_{t=1}^T (X_{i,t} - X_{\sourceset_t, t}) \\
& = \max_{1 \le j \le n} \max_{1 \le i \le n} \E_\sourceset \E_j \sum_{t=1}^T (X_{i,t} - X_{\sourceset_t, t}) \\
& \ge \max_{1 \le i \le n} \E_\sourceset \E_i \sum_{t=1}^T (X_{i,t} - X_{\sourceset_t, t}) \\
& \ge \E_\sourceset \left[\frac{1}{n} \sum_{i=1}^n \E_i \sum_{t=1}^T (X_{i,t} - X_{\sourceset_t, t})\right], \\
\end{align*}
where we have used the fact that the maximum is at least as large as the average in the final inequality. Since any player strategy in $\mathscr{P}$ lies in the convex hull of deterministic player strategies, a uniform bound~\eqref{EqnRegInter} over $\mathscr{P}_d$ implies that inequality~\eqref{EqnRegFinal} holds, as well.

We now turn to the proof of inequality~\eqref{EqnRegInter}. The idea is to show that on average, the player incurs a certain loss whenever the wrong source vertex is played, and this event must happen sufficiently often. We first write
\begin{align*}
\expect_i \sum_{t = 1}^{T} (X_{i, t} - X_{\sourceset_t, t}) & = \sum_{t = 1}^{T} \expect_i
\left[
\sum_{j \neq i} (X_{i, t} - X_{j, t}) \ind_{\{\sourceset_t = \{j\}\}}
\right] \\
& =
\sum_{j \neq i} \sum_{t = 1}^{T}
\expect_i
\left[X_{i, t} - X_{j, t}\right] \E_i\left[\ind_{\sourceset_t = \{j\}}\right].
\end{align*}
In the last equality, we have used the assumption that the adversary's action at each time is independent of the past to conclude that the difference in rewards $X_{i,t} - X_{j,t}$ (which depends on the adversary's action at time $t$) is independent of the indicator $\ind_{\sourceset_t = \{j\}}$ (which depends on the sequence of feedback received up to time $t-1$). Using the bound~\eqref{EqnWrongSource}, it follows that
\begin{align*}
\expect_i \sum_{t = 1}^{T} (X_{i, t} - X_{\sourceset_t, t}) = \sum_{j \neq i} r \E_i[T_j],
\end{align*}
where \(T_{i} = |\{t: \sourceset_t = \{i\}\}|\) denotes the number of times vertex \(i\) is selected as the source.

Now let $U_T$ denote a vertex drawn according to the distribution \(q_{T} = (q_{1, T}, \ldots, q_{n, T})\), where $q_{i, T} = \frac{T_{i}}{T}$. The derivations above imply that
\begin{align*}
\expect_i \sum_{t = 1}^{T} (X_{i, t} - X_{\sourceset_t, t}) =
rT \sum_{j \neq i} 
\prob_i\left\{U_{T} = j \right\} = rT \left(1 - \prob_i\left\{U_{T} = i \right\}\right),
\end{align*}
so
\begin{equation}
\frac{1}{n} \sum_{i = 1}^{n} \expect_i \sum_{t = 1}^{T} (X_{i, t} - X_{\sourceset_t, t})
=
rT
\left(1 - \frac{1}{n} \sum_{i = 1}^{n} \prob_i\left\{U_{T} = i \right\}
\right).
\label{eqn: clique lower 01}
\end{equation}

By Pinsker's inequality, we have
\begin{equation*}
\prob_i\left\{U_{T} = i \right\}
\leq 
\prob_0\left\{U_{T} = i \right\} 
+ \sqrt{\frac{1}{2} KL\left(\prob'_0, \prob'_i \right)},
\label{eqn: clique pinsker}
\end{equation*}
where $\mprob'_i$ denotes the distribution of $U_T$ under the adversarial strategy $\advset^i$. By Jensen's inequality, we therefore have
\begin{align}
\frac{1}{n} \sum_{i = 1}^{n} \prob_i\left\{U_{T} = i \right\} \leq 
\frac{1}{n} + \frac{1}{n} \sum_{i = 1}^{n} \sqrt{\frac{1}{2} KL\left(\prob'_0, \prob'_i \right)} \leq 
\frac{1}{n} 
+ \sqrt{\frac{1}{2n} \sum_{i = 1}^{n} KL\left(\prob'_0, \prob'_i \right)}.
\label{eqn: clique pinsker 02}
\end{align}
Finally, the chain rule for KL divergence implies that
\begin{equation}
\label{EqnKLInter}
KL\left(\prob'_0, \prob'_i\right)
=
KL\left(\prob_0, \prob_i\right)
+ 
\sum_{\mathscr{I}^{T}} 
\prob_0 \left\{\mathscr{I}^{T}\right\}
KL \left( \prob'_0\{\cdot | \mathscr{I}^{T}\} , \prob'_i \{\cdot | \mathscr{I}^{T}\} \right).
\end{equation}
Note that conditional on \(\mathscr{I}^{T}\), the distribution of \(U_{T}\) is the same under 
\( \prob'_0\) and \(\prob'_i\), since the player uses a deterministic strategy. Thus, equation~\eqref{EqnKLInter} implies that
\begin{equation}
\sum_{i=1}^n KL\left(\prob'_0, \prob'_i\right)
=
\sum_{i=1}^n KL\left(\prob_0, \prob_i\right)
\leq D.
\label{eqn: general kl final}
\end{equation}
Combining inequalities~\eqref{eqn: clique lower 01}, \eqref{eqn: clique pinsker 02}, and \eqref{eqn: general kl final}, we arrive at the desired result~\eqref{EqnRegInter}.
\end{proof}

To prove Theorems~\ref{theorem: complete graph lower} and \ref{theorem: directed lower bound}, it thus remains to find an appropriate set of strategies $\{\advset^0, \advset^1, \dots, \advset^n\}$ and verify the bounds~\eqref{EqnWrongSource} and~\eqref{EqnKLbound}. Details for the proofs are provided in Appendix~\ref{AppLBthms}.

%
\section{Additional online upper bound proofs}
\label{appAdditionalOnlineUpperBoundProofs}

In this Appendix, we provide proofs for the pseudo-regret of player strategies based on the OSMD algorithm. 
We begin with some preliminaries.

\subsection{Preliminaries}

We first describe the function $F_\psi$. Recall that a continuous function \(F: \overline{\mathcal{D}} \to \rr\) is a Legendre function if 
\(F\) is strictly convex, \(F\) has continuous first partial derivatives on \(\mathcal{D}\),
and 
\begin{equation*}
\lim_{x \to \overline{\mathcal{D}} \setminus \mathcal{D}} \|\nabla F(x) \| = \infty.
\end{equation*}
The analysis in this paper concerns a very specific type of Legendre function associated to a 0-potential, as described in the following definition:

\begin{definition*}
A function \(\psi: (-\infty, a) \to \rr_{+}\) is called a \(0\)-potential if it is convex, continuously differentiable, and satisfies the following conditions:
\begin{align*}
  &\begin{aligned}
\lim_{x \to - \infty} \psi(x) = 0, & &
\lim_{x \to a}  \psi(x) = \infty,  \hspace{24pt}\\ 
\psi' > 0,   \hspace{38pt} & \qquad & 
\int_{0}^{1} |\psi^{-1}(s)| ds \leq \infty.
  \end{aligned}
\end{align*}
We additionally define the associated function \(F_{\psi}\) on \((0, \infty)^{n}\) by
\[
F_{\psi}(x)
=
\sum_{i = 1}^{n} \int_{0}^{x_{i}} \psi^{-1}(s) ds.
\]
\label{definition: 0 potential}
\end{definition*}

In particular, we will consider the 0-potential \(\psi(x) = (- x)^{-q}\). Then \(\psi^{-1}(x) = - x^{- \frac{1}{q}}\), so
\[
F_{\psi}(x)
=
- \frac{q}{q - 1} \sum_{i = 1}^{n} x_{i}^{\frac{q - 1}{q}}.
\]
Specifically, we will consider the case $q = 2$ (the same analysis could be performed with respect to $q > 1$, and then the final bound could be optimized over $q$).

To employ Proposition~\ref{proposition: generic loss osmd}, we need to bound two summands.
The following simple lemma bounds the first term:

\begin{lemma*}
When \(\psi(x) = \frac{1}{x^2}\), we have the bound
\[
F_{\psi}(p) - F_{\psi}(p_{1})
\leq 2 \sqrt{n}, \qquad \forall p \in \Delta^n.
\]

\label{lemma: psi choice}
\end{lemma*}

\begin{proof}{Proof.}
Since $F_\psi(p) \le 0$ and \(\|p_{1}\|_{1} = 1\), H\"{o}lder's inequality implies that
\[
F_{\psi}(p) - F_{\psi}(p_{1})
\leq 2 \sum_{i=1}^n p_{1,i}^{1/2} \leq 2 n^{\frac{1}{2}}.
\]
This completes the proof of the lemma.
\end{proof}

All that remains is to analyze the loss-specific term appearing in Proposition~\ref{proposition: generic loss osmd} and  choose \(\eta\) appropriately.

\subsection{Proof of Theorem~\ref{theorem: symmetric loss osmd}}
\label{AppThmSymOSMD}

We first prove the following lemma:

\begin{lemma*}
\label{LemCondExp}
We have the inequality 
\begin{equation}
\sum_{i = 1}^{n}
\expect\left[\frac{(\lhat^{\text{sym}}_{i, t})^{2}}{(\psi^{-1})' (p_{i, t})}\right]
\leq \sqrt{2n}, \qquad \forall 1 \le t \le T.
\label{eqn: symmetric loss osmd 00}
\end{equation}
\end{lemma*}

\begin{proof}{Proof.}
Let $\cF_t$ denote the sigma-field of all actions up to time $t$. We have
\begin{align}
\sum_{i = 1}^{n}
\expect\left[\frac{(\lhat^{\text{sym}}_{i, t})^{2}}{(\psi^{-1})' (p_{i, t})}
\biggr| \mathcal{F}_{t - 1} \right]
& \stackrel{(a)}{=} 2\sum_{i = 1}^{n} p_{i, t}^{3/2} 
\expect\left[(\lhat_{i,t}^{\text{sym}})^{2} 
| \mathcal{F}_{t - 1}\right] \notag \\
& \stackrel{(b)}{\le} 2 \left(\sum_{i=1}^n p_{i,t}\right)^{1/2} \left(
\sum_{i = 1}^{n} \left(p_{i, t}\expect\left[(\lhat_{i,t}^{\text{sym}})^{2} 
| \mathcal{F}_{t - 1}\right] \right)^2
\right)^{1/2} \notag \\
& = 2 \left(
\sum_{i = 1}^{n} \left(p_{i, t}\expect\left[(\lhat_{i,t}^{\text{sym}})^{2} 
| \mathcal{F}_{t - 1}\right] \right)^2
\right)^{1/2},
\label{eqn: symmetric loss osmd 02}
\end{align}
where we have used the facts that \((\psi^{-1})'(x) = \frac{1}{2} x^{-3/2}\) and \(p_{t}\) is measurable with respect to \(\mathcal{F}_{t - 1}\) to establish $(a)$, and applied H\"{o}lder's inequality to obtain $(b)$.

We now inspect the conditional expectation more closely. 
We have
\begin{align*}
  &\begin{aligned}
\expect\left[(\lhat_{i,t}^{\text{sym}})^{2} 
| \mathcal{F}_{t - 1}\right] 
&= 
\expect\left[
\left(
\frac{1}{n} \sum_{j \neq i} \frac{1}{p_{i, t} + p_{j, t}} 
\ell_{i, j}^t Z_{ij}
\right)^{2} 
\biggr| \mathcal{F}_{t - 1}\right], \\ 
&=
\frac{1}{n^{2}}
\E\left[\sum_{j \neq i} \sum_{k \neq i}
\frac{1}{(p_{i, t} + p_{j, t}) (p_{i, t} + p_{k, t})}
\ell_{i, j}^t \ell_{i, k}^t Z_{ij} Z_{ik} \biggr| \cF_{t-1}\right] \\ 
& = \frac{1}{n^2} \E\left[\sum_{j \neq i} \sum_{k \neq i} \frac{1}{(p_{i,t} + p_{j,t})(p_{i,t} + p_{k,t})} \ell_{i,j}^t \ell_{i,k}^t Z_i \biggr| \cF_{t-1}\right] \\
&\qquad + \frac{1}{n^2} \E\left[\sum_{j \neq i} \frac{1}{(p_{i,t} + p_{j,t})^2} (\ell_{i,j}^t)^2 Z_j \biggr| \cF_{t-1}\right],
  \end{aligned}
\end{align*}
where the third equality is due to the fact that \(Z_{ij} Z_{ik}\) is \(1\) only when \(i\) is the source vertex or \(j = k\) is the source vertex. Using the fact that \(\ell_{i,j}^t\) is bounded by \(1\), we then obtain
\begin{align*}
  &\begin{aligned}
\expect\left[(\lhat_{i,t}^{\text{sym}})^{2} 
| \mathcal{F}_{t - 1}\right] & \le \frac{1}{n^2} \E\left[\sum_{j \neq i} \sum_{k \neq i} \frac{Z_i}{(p_{i,t} + p_{j,t})(p_{i,t} + p_{k,t})} \biggr| \cF_{t-1}\right] \\
&\qquad + \frac{1}{n^2} \E\left[\sum_{j \neq i} \frac{Z_j}{(p_{i,t} + p_{j,t})^2} \biggr| \cF_{t-1}\right] \\\
&\leq 
\frac{1}{n^{2}}
\sum_{j \neq i} \sum_{k \neq i}
\frac{p_{i, t}}{(p_{i, t} + p_{j, t}) (p_{i, t} + p_{k, t})} +
\frac{1}{n^{2}}
\sum_{j \neq i}
\frac{p_{j, t}}{(p_{i, t} + p_{j, t})^{2}} \\ 
&\leq 
\frac{1}{n^{2}}
\sum_{j \neq i} \sum_{k \neq i}
\frac{p_{i, t} + p_{j, t}}{(p_{i, t} + p_{j, t}) (p_{i, t} + p_{k, t})} \\ 
&=
\frac{1}{n^{2}} \sum_{j \neq i} \sum_{k \neq i}
\frac{1}{p_{i, t} + p_{k, t}} \\ 
&\le
\frac{1}{n} \sum_{k=1}^n
\frac{1}{p_{i, t} + p_{k, t}}.
  \end{aligned}
\end{align*}
Combining this result with the bound~\eqref{eqn: symmetric loss osmd 02}, we have
\begin{align*}
\sum_{i = 1}^{n}
\expect\left[\frac{(\lhat^{\text{sym}}_{i, t})^{2}}{(\psi^{-1})' (p_{i, t})}
\biggr| \mathcal{F}_{t - 1} \right]
&\leq 
2\left(\sum_{i = 1}^{n} \left( \frac{p_{i, t}}{n} \sum_{k = 1}^{n} \frac{1}{p_{i, t} + p_{k, t}}\right)^2 \right)^{1/2} \\
&=
\frac{2}{n}
\left(\sum_{i = 1}^{n}
\sum_{j = 1}^{n} \sum_{k = 1}^{n}
\frac{p_{i, t}^{2}}{(p_{i, t} + p_{j, t}) (p_{i, t} + p_{k, t})}
\right)^{\frac{1}{2}} \\ 
&\leq 
\frac{2}{n}
\left(n \sum_{i = 1}^{n} \sum_{j = 1}^{n} \frac{p_{i, t}}{p_{i, t} + p_{j, t}}
\right)^{\frac{1}{2}}. 
\end{align*}
Now, we have the useful equation
\begin{equation}
\label{EqnDoubleSum}
\sum_{i = 1}^{n} \sum_{k =1}^{n}
\frac{a_{i}}{a_{i} + a_{k}}
=
\frac{n^{2}}{2},
\end{equation}
for any nonnegative sequence $\{a_i\}_{i=1}^n$. This may be seen via the following algebraic manipulations:
\begin{align*}
  &\begin{aligned}
\sum_{i = 1}^{n} \sum_{k = 1}^{n}
\frac{a_{i}}{a_{i} + a_{k}}
&=
\sum_{i = 1}^{n} \frac{a_{i}}{a_{i} + a_{i}}
+
\sum_{k \neq i} \frac{a_{i}}{a_{i} + a_{k}} \\ 
&=
\frac{n}{2}
+ \frac{1}{2} \sum_{k \neq i} \frac{a_{i}}{a_{i} + a_{k}}
+ \frac{1}{2} \sum_{k \neq i} \frac{a_{k}}{a_{i} + a_{k}} \\ 
&=
\frac{n}{2} + \frac{1}{2} \sum_{k \neq i} \frac{a_i + a_k}{a_i + a_k} \\
&=
\frac{n}{2}
+ \frac{n (n - 1)}{2} \\ 
&=
\frac{n^{2}}{2}.
    \end{aligned}
\end{align*}
Appealing to equation~\eqref{EqnDoubleSum}, we may replace the double sum by $\frac{n^2}{2}$ and simplify the bound:
\begin{align*}
  &\begin{aligned}
\sum_{i = 1}^{n}
\expect\left[\frac{(\lhat^{\text{sym}}_{i, t})^{2}}{(\psi^{-1})' (p_{i, t})}
\biggr| \mathcal{F}_{t - 1} \right]
&\leq 
\frac{2}{n}
\left(\frac{n^{3}}{2}\right)^{\frac{1}{2}}
=
\sqrt{2 n}.
  \end{aligned}
\end{align*}
Taking an additional expectation and using the tower property, we arrive at the desired inequality. 
\end{proof}

Combining Lemmas~\ref{lemma: psi choice} and~\ref{LemCondExp} with Proposition~\ref{proposition: generic loss osmd}, we then have
\[
\sup_{\advset \in \mathscr{A}} \overline{R}_{T}(\advset, \sourceset_{\text{OSMD}}^{\text{sym}})
\leq 
\frac{2 \sqrt{n}}{\eta} + \eta T \sqrt{\frac{n}{2}}.
\]
Optimizing over $\eta$, we take \(\eta = 2^{\frac{3}{4}} T^{-\frac{1}{2}}\), which establishes the desired bound.

\section{Adversarial influence maximization with multiple sources}
\label{AppMultipleSources}

In this Appendix, we prove results concerning multiple influence sources. 
First, we need to give the precise algorithmic details of the online greedy algorithm.
We assume the player is allowed to choose source vertices sequentially at time $t$ and observes the corresponding edge feedback immediately after each selection. The algorithm, inspired by \cite{streeter2007}, is outlined below:

\begin{framed}
\noindent \textbf{Online Greedy Algorithm} \\

\noindent
Given: A single-source player strategy $\sourceset^1$. \\
Output: A $k$-source player strategy $\sourceset^k = \{\sourceset_t\}_{1 \le t \le T}$. \\

\noindent
For each \(t = 1, \ldots, T\), choose \(\sourceset_{t} = \{v_{1, t}, \ldots, v_{k, t}\}\) sequentially, as follows:

\begin{enumerate}
\item[(1)] Select $v_{1,t}$ according to the single-source strategy $\sourceset^1$.

\item[(2)] For each \(i > 1\), select $v_{i,t}$ according to the single-source strategy $\sourceset^1$, based on the edge feedback
\(\mathscr{I}(\advset_{t}, \{v_{1, t}, \ldots, v_{i, t}\}) \setminus \mathscr{I}(\advset_{t}, \{v_{1, t}, \ldots, v_{i - 1, t}\})\).

\end{enumerate}
\end{framed}

In other words, the Online Greedy Algorithm runs the player's strategy for single-source selection $k$ times in parallel, with losses computed marginally for each successively chosen vertex. The ``greedy'' component of the algorithm corresponds to the fact that the player makes a selection of the set of $i^\text{th}$ source vertices in the best possible way based on the information available (i.e., according to the single-source strategy that is designed to incur a small pseudo-regret). Note that the feedback
\begin{equation*}
\mathscr{I}(\advset_{t}, \{v_{1, t}, \ldots, v_{i, t}\}) \setminus \mathscr{I}(\advset_{t}, \{v_{1, t}, \ldots, v_{i - 1, t}\})
\end{equation*}
is indeed computable by the player when choosing the $i^{\text{th}}$ vertex at round $t$, since the player has already observed $\mathscr{I}(\advset_{t}, \{v_{1, t}, \ldots, v_{i - 1, t}\})$ after the first $i-1$ source nodes are selected.

Fix an adversarial strategy $\advset$, and define the functions
\(f_{t}(\sourceset_{t}) = f(\advset_{t}, \sourceset_{t})\) and \(F(\sourceset) = \sum_{t = 1}^{T}f_{t}(\sourceset_{t})\). Thus, $F(\sourceset)$ is the total reward for strategy $\sourceset = \{\sourceset_t\}$.
In the stochastic setting, when \(T = 1\), many influence-maximization analyses exploit the submodularity of \(f_{t}\) under certain stochastic assumptions on \(\advset_{t}\).
In the bandit setting, we wish to establish an analogous result for \(F\), in order to establish regret bounds when the player chooses source vertices according to a greedy algorithm.

Since $\sourceset \in (2^V)^T$, the function $F$ is not technically a set function. However, we may identify each player strategy $\sourceset$ with an element of $\sourceset^* \in 2^{(\vertexset^T)}$, and define $F^*(\sourceset^*) = F(S)$. Here,
\begin{equation*}
V^T := \left\{v^T = \big(v(1), v(2), \dots, v(T)\big) \mid v(i) \in V, \text{ for } 1 \leq i \leq T\right\},
\end{equation*}
and $\sourceset^* = \{u_1^T, \dots, u_k^T\} \in 2^{(V^T)}$ corresponds to the strategy that selects the source nodes $\{u_1(t), \dots, u_k(t) \}$ in round $t$.

In more detail, let \(S_{t}^{(i)} = \{s_{t}(1), \ldots, s_{t}(i)\} \subseteq \vertexset\) denote the set of the first \(i\) seed vertices in round \(t\), where \(S_{t}^{(0)} = \emptyset\).
Then, we can write 
\[
f_{t}(S_{t}) 
=
\sum_{i = 1}^{k} f_{t}\left(S_{t}^{(i)}\right) - f_{t}\left(S_{t}^{(i - 1)}\right).
\]
One can then write the total reward as 
\begin{align*}
&\begin{aligned}
F(S) 
&= 
\sum_{t = 1}^{T} \sum_{i = 1}^{k} f_{t}\left(S_{t}^{(i)}\right) - f_{t}\left(S_{t}^{(i - 1)}\right) 
=
\sum_{i = 1}^{k} \sum_{t = 1}^{T}  f_{t}\left(S_{t}^{(i)}\right) - f_{t}\left(S_{t}^{(i - 1)}\right)
\end{aligned}
\end{align*}
If we define
\begin{align*}
& \begin{aligned}
f_{i}^{*}(S^{*}) 
&=
F^{*}(\{u_1^T, \dots, u_i^T\}) - F^{*}(\{u_1^T, \dots, u_{i - 1}^T\}) \\
&= 
\sum_{t = 1}^{T}
f_{t}\left(\{u_{j}(t): j \leq i \}\right)
-
f_{t}\left(\{u_{j}(t): j \leq i - 1 \}\right)
   \end{aligned}
\end{align*}
and \(F^{*}(S^{*}) = \sum_{i = 1}^{k} f_{i}^{*}(S^{*})\),
then we indeed get the desired equality \(F(S) = F^{*}(S^{*})\) while also switching our summation for submodularity to be over the \(i\)th vertices as opposed to the \(t\)th round.

We first show that $F^*$ is a monotone, submodular function:
\begin{lemma*}
The function 
\(f_{t}(\sourceset_{t}) = f(\advset_{t}, \sourceset_{t})\)
is monotone and submodular, for every fixed \(\advset_{t}\).
\label{lemma: little f submodular}
\end{lemma*}

\begin{proof}{Proof.}
It is trivial to see that $f_t$ is monotone, so we focus on proving submodularity. Our goal is to show that for a fixed \(\advset_{t}\), and for any $\sourceset_t \subseteq \sourceset'_t$ and $u \in \vertexset \backslash \sourceset'_t$, we have
\begin{equation}
\label{EqnSubmodVerify}
f_t(\sourceset'_t \cup \{u\}) - f_t(\sourceset'_t) \le f_t(\sourceset_t \cup \{u\}) - f_t(\sourceset_t).
\end{equation}
Let \(Z_{\sourceset_{t}, v}\) denote the indicator of an open path between a source node \(s \in \sourceset_{t}\) and \(v \in \vertexset\), where by convention, $Z_{\sourceset_t, v} = 1$ if $v \in \sourceset_t$. Note that $f_t(\sourceset_t) = \frac{1}{n} \sum_{v \in \vertexset} Z_{\sourceset_t, v}$. We will show that for $v \notin \sourceset'_t \cup \{u\}$, we have
\begin{equation}
Z_{\sourceset_{t}' \cup \{u\}, v} - Z_{\sourceset_{t}', v} \leq  Z_{\sourceset_{t} \cup \{u\}, v} - Z_{\sourceset_{t}, v}.
\label{eqn: little f lemma}
\end{equation}
Summing over $v \in (\sourceset'_t \cup \{u\})^c$, using $Z_{\sourceset_{t} \cup \{u\}, v} - Z_{\sourceset_{t}, v} \geq 0$ for $v \in \sourceset_t' \setminus \sourceset_t$, and dividing by $n$ will yield the desired inequality~\eqref{EqnSubmodVerify}.

We have three cases to consider: In the first case, an open path exists from some \(s \in \sourceset_{t}'\) to \(v\).
Then the left side of inequality~\eqref{eqn: little f lemma}
is equal to \(0\), while the right hand side is at least \(0\) by monotonicity.
In the second case, an open path does not exist from any \(s \in \sourceset_{t}'\) to \(v\), but an open path exists from \(u\) to \(v\).
Then both sides of inequality~\eqref{eqn: little f lemma} are equal to \(1\).
Finally, if no open path exists from 
\(s \in \sourceset_{t}' \cup \{u\}\) to \(v\), 
then both sides of inequality~\eqref{eqn: little f lemma} are equal to \(0\).
This completes the proof.
\end{proof}

\begin{proposition*}
The function \(F^{*}\) is monotone and submodular.
\label{prop: big f submodular}
\end{proposition*}

\begin{proof}{Proof.}
The properties are essentially immediate from Lemma~\ref{lemma: little f submodular}. Let \(\mathcal{P}\) and \(\mathcal{Q}\) be elements of 
\(\left(2^{\vertexset}\right)^{T}\)
such that 
\(\mathcal{P}^{*} \subseteq \mathcal{Q}^{*}\). Then
\begin{align*}
&\begin{aligned}
F^{*}(\mathcal{P}^{*})
=
\sum_{t = 1}^{T} 
f_{t}(\mathcal{P}_{t})  
\leq 
\sum_{t = 1}^{T}
f_{t}(\mathcal{Q}_{t})
= 
F^{*}(\mathcal{Q}^{*}),
\end{aligned}
\end{align*}
proving monotonicity. Similarly, if $\sourceset \in (2^\vertexset)^T$, we have
\begin{align*}
&\begin{aligned}
F^{*}(\sourceset^{*} \cup \mathcal{Q}^{*}) - F^{*}(\mathcal{Q}^{*})
&=
\sum_{t = 1}^{T} 
\left(f_{t}(\sourceset_{t} \cup \mathcal{Q}_{t})
-
f_{t}(\mathcal{Q}_{t})\right) \\ 
&\leq 
\sum_{t = 1}^{T}
\left(f_{t}(\sourceset_{t} \cup \mathcal{P}_{t})
-
f_{t}(\mathcal{P}_{t})\right) \\ 
&= 
F^{*}(\sourceset^{*} \cup \mathcal{P}^{*}) - F^{*}(\mathcal{P}^{*}),
\end{aligned}
\end{align*}
proving submodularity.
\end{proof}

By the standard greedy approximation (\cite{KemEtal03, NemEtal78}), we then have
\[
\left(1 - \frac{1}{e}\right) \max_{|\sourceset^{*}| \leq K} F^{*}(\sourceset^{*})
\leq 
F^{*}(G^{*}),
\]
where \(G^{*}\) is a set of cardinality $K \ge 1$ constructed via a sequential greedy algorithm.
However, this result is not immediately applicable to the online bandit setting, since we do not have direct access to \(F^{*}\). Thus, we can only hope to obtain an approximate greedy maximizer $\Gtilde^*$, and we wish to derive theoretical guarantees for $F^*(\Gtilde^*)$.

Our result relies on the following general proposition:
\begin{proposition*}[Theorem 6 from \cite{streeter2007}]
Let \(f: 2^{\mathscr{V}} \to \rr\) be a monotone, submodular function such that 
\(f(\emptyset) = 0\).
Consider a set \(\mathscr{D} \subseteq \mathscr{V}\) and a sequence of error tolerances $\{\epsilon_i\}$, and suppose \(\{G^\epsilon_{i}\}\) is constructed in an approximate greedy manner, such that $G_0^\epsilon = \emptyset$ and $G_i^\epsilon = G_{i-1}^\epsilon \cup \{g_i\}$, where
\begin{equation*}
\max_{d \in \mathscr{D}} f(G^{\epsilon}_{i - 1} \cup \{d\})  
- f(G^{\epsilon}_{i - 1})
\leq 
f(G^{\epsilon}_{i - 1} \cup \{g_{i}\}) - f(G^{\epsilon}_{i - 1}) + \epsilon_{i}.
\end{equation*}
Then for any $K \ge 1$, we have
\[
\left(1 - \frac{1}{e}\right) 
\max_{\sourceset^{*} \in \mathscr{D}_{K}} 
f(\sourceset^{*}) 
- f(G^{\epsilon}_{K})
\leq 
\sum_{i = 1}^{K} \epsilon_{i},
\]
where \(\mathscr{D}_{K}\) consists of subsets of $\mathscr{D}$ containing at most \(K\) elements.
\label{theorem: 6 from streeter}
\end{proposition*}

Proposition~\ref{theorem: 6 from streeter} ensures that for submodular functions,
successive errors \(\{\epsilon_{i}\}\) in a sequential greedy algorithm only accumulate additively. The proof is provided in \cite{streeter2007}, but we include a proof in Appendix~\ref{AppStreeterProof} for completeness.

\subsection{Proof of Theorem~\ref{theorem: symmetric loss multisource}}

Suppose $\advset \in \mathscr{A}$. We will apply Proposition~\ref{theorem: 6 from streeter} with $f = \expect_{\advset}F^*$, $\mathscr{V} = \vertexset^T$, and $K = k$. Note that $\E_\advset F^*$ inherits monotonicity and submodularity from $F^*$. 
Also let
\begin{equation*}
\mathscr{D} = \{(v, \dots, v): v \in \vertexset\} \subseteq 2^{(\vertexset^T)}
\end{equation*}
denote the diagonal set of $2^{(\vertexset^T)}$. For a (non-random) $k$-source strategy $\sourceset^*$ with $|\sourceset_t^*| = k$ for all $t$, we use the notation $\sourceset^{*} = \{\sourceset^{*}_{1}, \ldots, \sourceset^{*}_{k}\}$, where \(\sourceset^{*}_{i}\) corresponds to the set of $i^\text{th}$ vertices chosen during the $T$ rounds. 
Proposition~\ref{theorem: 6 from streeter} immediately gives
\begin{align*}
  &\begin{aligned}
\left(1 - \frac{1}{e}\right) & \max_{\sourceset^{*} \in \mathscr{D}} \expect_{\advset} F^{*}(\sourceset^{*}) - \expect_{\advset} F^{*} (G^{\epsilon}_{k}) \\
 &\leq 
\sum_{i = 1}^{k}
\max_{d_i \in \mathscr{D}} \expect_{\advset}[ F^{*}(G_{i - 1}^{\epsilon} \cup \{d_i\}) 
- F^{*}(G^{\epsilon}_{i - 1} \cup \{g_{i}\})],
  \end{aligned}
\end{align*}
where the sets $\{G^\epsilon_i\}$ are chosen in an approximate greedy manner, and  \(\epsilon_{i}\) are upper bounded by the regret for the \(i\)th instance of the single-source algorithm. In particular, we consider $\{G^\epsilon_i\}$ to be the choice of $i^\text{th}$ vertices $\sourceset^*_i$ corresponding to the player's choice under the strategy $\sourceset^1$.

We now take an expectation with respect to possible randomization in the player's strategy, to obtain
\begin{align*}
  &\begin{aligned}
\overline{\regret}_{T}^{(1 - 1/e)}(\advset, \sourceset) 
&\leq 
\sum_{i = 1}^{k}
\expect_{\sourceset}\left[\max_{d_i \in \mathscr{D}} 
\expect_{\advset} \left[ F^{*}(G_{i - 1}^{\epsilon} \cup \{d_i\}) 
-  F^{*}(G^{\epsilon}_{i - 1} \cup \{g_{i}\})\right]
\right]\\
&\stackrel{(a)}= \sum_{i = 1}^{k}
\expect_{\sourceset_{[1:i]}}\left[\max_{d_i \in \mathscr{D}} 
\expect_{\advset} \left[ F^{*}(G_{i - 1}^{\epsilon} \cup \{d_i\}) 
-  F^{*}(G^{\epsilon}_{i - 1} \cup \{g_{i}\})\right]
\right]\\
&\stackrel{(b)}=
\sum_{i = 1}^{k}
\expect_{\sourceset_{[1:i-1]}}\left[\E_{\sourceset_i}\max_{d_i \in \mathscr{D}} 
\expect_{\advset} \left[ F^{*}(G_{i - 1}^{\epsilon} \cup \{d_i\}) 
-  F^{*}(G^{\epsilon}_{i - 1} \cup \{g_{i}\})\right]
\right].
  \end{aligned}
\end{align*}
Here, $\E_{\sourceset_{[1:i]}}$ denotes the expectation with respect to the first $i$ vertices played, and the equality in $(a)$ holds because the set of $i^{\text{th}}$ vertices played depends only on the sets of the first $i$ vertices played. The equality in $(b)$ holds because the set \(G_{i - 1}\), and hence the choice of \(d_i\), does not depend on the selection of \(i^\text{th}\) vertices.
Furthermore, the inner expression is simply the pseudo-regret of strategy $\sourceset^1$. By Theorem~\ref{theorem: symmetric loss osmd}, this is bounded by \(2^{\frac{1}{4}} \sqrt{T n}\). Summing up, we obtain the desired result.

%
\subsection{Proof of Proposition~\ref{theorem: 6 from streeter}}
\label{AppStreeterProof}

We begin with two supporting lemmas:

\begin{lemma*}
For any 
\(\mathcal{P} \subseteq \mathscr{V}\) and 
\(\mathcal{Q} \subseteq \mathscr{D}\),
we have
\[
f(\mathcal{P} \cup \mathcal{Q})
\leq 
f(\mathcal{P}) + |\mathcal{Q}| \max_{v \in \mathscr{D}} [f(\mathcal{P} \cup \{v\}) - f(\mathcal{P})].
\]

\label{lemma: submodularity over time}
\end{lemma*}

\begin{proof}{Proof.}
We proceed by induction on \(|\mathcal{Q}|\).
The case \(|\mathcal{Q}| = 1\) is immediate.
Now suppose the statement is true for all \(|\mathcal{Q}| \leq k\), where $k \ge 1$.
Let \(c \in \mathscr{D}\), and suppose \(\mathcal{Q} \subseteq \mathscr{D}\) has cardinality $k$.
Then
\begin{align*}
&\begin{aligned}
f\left(\mathcal{P} \cup (\mathcal{Q} \cup \{c\})\right)
&\stackrel{(a)}{\leq} 
f\left(\mathcal{P} \cup \{c\}\right) +  |\mathcal{Q}| \max_{d \in \mathscr{D}} [f\left((\mathcal{P} \cup \{c\}\right) \cup \{d\}) - f\left(\mathcal{P} \cup \{c\}\right)] \\ 
&\stackrel{(b)}{\leq} 
f(\mathcal{P}) +  \max_{d \in \mathscr{D}} [f\left(\mathcal{P} \cup \{d\}\right) - f(\mathcal{P})]
+ |\mathcal{Q}| \max_{d \in \mathscr{D}} [f\left(\mathcal{P} \cup \{d\}\right) - f(\mathcal{P})] \\ 
&= 
f(\mathcal{P}) + |\mathcal{Q} \cup \{c\}| \max_{d \in \mathscr{D}} [f\left(\mathcal{P} \cup \{d\}\right) - f(\mathcal{P})],
\end{aligned}
\end{align*}
where $(a)$ follows from the induction hypothesis and $(b)$ follows from the induction hypothesis and submodularity.
This completes the induction and proves the lemma.
\end{proof}

\begin{lemma*}
Let $\delta_i :=  f(G^{\epsilon}_{i}) - f(G^{\epsilon}_{i - 1})$.
For any $\mathcal{Q} \subseteq \mathscr{D}$, we have
\[
f(\mathcal{Q})
\leq 
f(G^{\epsilon}_{i - 1}) + |\mathcal{Q}| (\delta_{i} + \epsilon_{i}).
\]

\label{lemma: fact 2}
\end{lemma*}

\begin{proof}{Proof.}
Using Lemma \ref{lemma: submodularity over time} and monotonicity of $f$, we have
\begin{align*}
&\begin{aligned}
f(\mathcal{Q})
&\leq 
f(G^{\epsilon}_{i - 1} \cup \mathcal{Q}) \\ 
&\leq 
f(G^{\epsilon}_{i - 1}) + |\mathcal{Q}| \max_{d \in \mathscr{D}} [f(G^{\epsilon}_{i - 1} \cup \{d\}) - f(G^{\epsilon}_{i - 1})] \\ 
&\leq 
f(G^{\epsilon}_{i - 1}) 
+ |\mathcal{Q}| \left(f(G^{\epsilon}_{i}) - f(G^{\epsilon}_{i - 1}) + \epsilon_{i}\right) \\ 
&= 
f(G^{\epsilon}_{i - 1}) + |\mathcal{Q}|(\delta_{i} + \epsilon_{i}),
\end{aligned}
\end{align*}
completing the proof.
\end{proof}

We now define 
\(\Delta_{i} 
:= \max_{\sourceset^{*} \in \mathscr{D}_{K}} 
f(\sourceset^{*}) - f(G^{\epsilon}_{i - 1})\).
By Lemma \ref{lemma: fact 2},
we have
\[
\max_{\sourceset^{*} \in \mathscr{D}_{K}} f(\sourceset^{*})
\leq 
f(G^{\epsilon}_{i - 1}) + K (\delta_{i} + \epsilon_{i}).
\]
Subtracting \(f(G^{\epsilon}_{i - 1})\), we obtain
\[
\Delta_{i}
\leq 
K(\delta_{i} + \epsilon_{i})
=
K(\Delta_{i} - \Delta_{i + 1} + \epsilon_{i}),
\]
so
\[
\Delta_{i + 1}
\leq 
\Delta_{i} \left(1 - \frac{1}{K}\right) + \epsilon_{i}.
\]
Applying this inequality recursively, we see that
\begin{align*}
\Delta_{K+1} \leq
\Delta_1 \prod_{i = 1}^{K}\left(1 - \frac{1}{K}\right) + \sum_{i = 1}^K \epsilon_{i} =
\Delta_1\left(1 - \frac{1}{K}\right)^K + \sum_{i = 1}^K \epsilon_{i} \leq 
\Delta_1\left(\frac{1}{e}\right) + \sum_{i = 1}^K \epsilon_{i}. 
\end{align*}
Rearranging and using the fact that \(f(\emptyset) = 0\) completes the proof.

\section{Additional online lower bound proofs}
\label{section: general lower}

The main goal of this Appendix is to prove Theorems~\ref{theorem: complete graph lower} and \ref{theorem: directed lower bound}.
Some of the computations are rather lengthy and are therefore included in Appendix~\ref{section: kl proofs}.

\subsection{Proofs of theorems}
\label{AppLBthms}

We first present the main components of the proofs, followed by detailed calculations involving the Kullback-Leibler divergence.

\subsubsection{Proof of Theorem~\ref{theorem: complete graph lower}}
\label{section: complete graph}

Let the adversarial strategies $\{\advset^i\}$ be defined as follows: For each strategy, the adversary chooses a random subset of vertices, and opens all edges between vertices in the subset. For $\advset^i$, with $1 \le i \le n$, the adversary includes vertex $i$ with probability \(\frac{c}{n}\), and includes all other vertices with probability \(\frac{c}{n}(1 - \delta)\) each, where $\delta \in (0,1/2)$ is a small constant. Finally, for $\advset^0$, the adversary includes all vertices independently with probability $\frac{c}{n}(1-\delta)$. Successive actions of the adversary are i.i.d.\ across time steps.

We now derive the following lemmas, which will be used in Proposition~\ref{PropLB}:

\begin{lemma*}
For any $i \neq j$ and $1 \le t \le T$, we have
\[
\expect_i [X_{i, t} - X_{j, t}]
=
\frac{(n - 2) c^{2}}{n^{3}} (1 - \delta) \delta.
\]

\label{lemma: clique regret}
\end{lemma*}

\begin{proof}{Proof.}
Let $\cC_t$ be the clique chosen by the adversary at time $t$. Note that if $i,j \in \cC_t$ or $i,j \notin \cC_t$, the difference in rewards is 0. Thus, the only cases of interest in computing the expectation are when exactly one of \(i\) or \(j\) is in $\cC_t$. Then
\begin{align*}
 & \begin{aligned}
\expect_i \left[X_{i, t} - X_{j, t}\right]
&= \E_i\left[(|\cC_t|-1) \ind_{i \in \cC_t} \ind_{j \notin \cC_t} - (1-|\cC_t|) \ind_{i \notin \cC_t} \ind_{j \in \cC_t}\right] \\
&=
\frac{1}{n}
\left( \frac{c}{n} (n - 2)(1 - \delta)\right)
\left(\frac{c}{n} \right)
\left[1 - \frac{c}{n}(1 - \delta)\right] \\ 
&\qquad 
- \frac{1}{n}
\left( \frac{c}{n} (n - 2) (1 - \delta)\right)
\left[1 - \frac{c}{n}\right]   
\left[\frac{c}{n}(1 - \delta)\right] \\ 
&= 
\frac{1}{n}(n - 2)  \left(\frac{c}{n}\right)^{2} (1 - \delta)
\left(
 \left[1 - \frac{c}{n}(1 - \delta) \right]
-
\left[1 - \frac{c}{n}\right] (1 - \delta)
\right) \\ 
&=
\frac{(n - 2)c^{2}}{n^{3}} (1 - \delta) 
\delta,
  \end{aligned}
\end{align*}
where the second equality uses the fact that \(\frac{c}{n} (n - 2)  (1 - \delta)\) other vertices are expected to be in $\cC_t$.
\end{proof}

\begin{lemma*}
Let $\sourceset \in \mathscr{P}_d$ be a deterministic player strategy, and let $T_i = |\{t: \sourceset_t = \{i\}\}|$. Then
we have the upper bound
\[
\sum_{i = 1}^{n} KL\left(\mprob_0, \mprob_i\right)
\leq 
\frac{c(c + 1)}{n - c} T \delta^{2}.
\]
\label{lemma: clique kl computation}
\end{lemma*}
The proof of Lemma~\ref{lemma: clique kl computation} is provided in Appendix~\ref{AppCliqueKL}.

Thus, by Proposition~\ref{PropLB}, we have
\begin{align*}
\inf_{\sourceset \in \mathscr{P}} \sup_{\advset \in \mathscr{A}} \ER(\advset, \sourceset)
&\geq 
T \frac{(n - 2)c^2}{n^3}(1 - \delta) \delta 
\left(
\frac{n - 1}{n}
- \delta \sqrt{\frac{T}{2n}}
\sqrt{\frac{c}{n - c}(c + 1)}
\right) \\
&\geq 
\frac{T}{6} \left(\frac{c}{n}\right)^{2}
\left(\frac{n - 1}{n} \delta - \delta^{2} \sqrt{\frac{T}{2n}} \sqrt{\frac{c (c + 1)}{n - c}}
\right),
\end{align*}
where the second inequality uses the fact that \(n \geq 3\) and \(\delta < 1/2\). Finally, we optimize over \(\delta\) and \(c\). Since we have a quadratic equation in \(\delta\), we take
\[
\delta 
=
\frac{n - 1}{2n} \sqrt{\frac{2n}{T}} \sqrt{\frac{n - c}{c (c + 1)}},
\]
yielding
\begin{align*}
\inf_{\sourceset \in \mathscr{P}} \sup_{\advset \in \mathscr{A}} \ER(\advset, \sourceset)
&\geq 
\frac{T}{6} \left(\frac{c}{n}\right)^{2}
\left(\frac{1}{4}  \left(\frac{n - 1}{n}\right)^{2} \sqrt{\frac{2n}{T}} \sqrt{\frac{n - c}{c (c + 1)}}\right) \\
&= 
\frac{1}{12 \sqrt{2}} \sqrt{T} \left(\frac{c}{n}\right)^{2} \left(\frac{n - 1}{n}\right)^{2} \sqrt{\frac{n (n - c)}{c (c+ 1)}} \\
& \geq 
\frac{1}{27 \sqrt{3}} \sqrt{T} \left(\frac{c}{n^2}\right) \sqrt{n (n - c)},
\end{align*}
where the second inequality uses the bounds $\frac{n-1}{n} \ge \frac{2}{3}$ when $n \ge 3$, and $\frac{c}{c+1} \ge \frac{2}{3}$ when $c \ge 2$. The final expression is optimized at \(c = \frac{2n}{3}\), yielding the desired lower bound. 
Note that for this choice of \(c\), we indeed have \(\delta < 1/2\) when \(T \geq 2\).

%
\subsubsection{Proof of Theorem~\ref{theorem: directed lower bound}}
\label{SecDirectedLB}

Let the adversarial strategies $\{\advset^i\}$ be defined as follows: For each strategy, the adversary independently designates every vertex to be a source, sink, or neither. The adversary then opens directed edges from all source vertices to all sink vertices. For $\advset^i$, with $1 \le i \le n$, the adversary designates vertex $i$ to be a source vertex with probability $\frac{c}{n}$, and all other vertices to be source vertices with probability $\frac{c}{n}(1-\delta)$. All vertices are designated to be sink vertices with probability $\frac{d}{n}$. Finally, for $\advset^0$, the adversary designates all vertices to be source vertices with probability $\frac{c}{n}(1-\delta)$, and sink vertices with probability $\frac{d}{n}$. Successive actions of the adversary are i.i.d.\ across time steps.

We now derive the following lemmas, which will be used in Proposition~\ref{PropLB}:

\begin{lemma*}
For any $i \neq j$ and $1 \le t \le T$, we have
\[
\expect_{i}
[X_{i, t} - X_{j, t}]
=
\frac{(n - 1) c d}{n^{3}} \delta.
\]
\label{lemma: directed regret}
\end{lemma*}

\begin{proof}{Proof.}
We compute the expectation of each term separately. Let $\cB_t$ and $\cC_t$ denote the source and sink vertices at time $t$, respectively.
Note that $X_{i,t} = \frac{1}{n}$ if $i \notin \cB_t$; otherwise, $X_{i,t} = \frac{1+|\cC_t|}{n}$. Hence,
\begin{align*}
n\expect_{i}[X_{i, t}]
& = \E\left[\ind_{i \notin \cB_t} + (1+|\cC_t|) \ind_{i \in \cB_t}\right] \\
& = \left(1- \frac{c}{n}\right) + \left(1+ (n-1)\frac{d}{n}\right) \left(\frac{c}{n}\right) \\
& = 1 +  \frac{(n - 1) cd}{n^{2}}.
\end{align*}
The computation for \(X_{j, t}\) is similar:
\begin{align*}
n\expect_{i}[X_{j, t}] 
&= \E\left[\ind_{j \notin \cB_t} + (1+|\cC_t|) \ind_{j \in \cB_t}\right] \\
& = \left(1-\frac{c}{n}(1-\delta)\right) + \left(1+(n-1)\frac{d}{n}\right) \frac{c}{n} (1-\delta) \\
& = 1 + \frac{(n - 1) cd}{n^{2}} (1 - \delta).
\end{align*}
Taking the difference between these expectations proves the lemma.
\end{proof}

%
\begin{lemma*}
Let $\sourceset \in \mathscr{P}_d$ be a deterministic player strategy, and let $T_i = |\{t: \sourceset_t = \{i\}\}|$. 
Then we have the upper bound
\[
\sum_{i = 1}^{n} KL\left(\prob_0, \prob_i\right)
\leq 
\frac{c ( n - d)}{n (n - c - d)} T
\delta^{2}. 
\]
\label{lemma: directed kl}
\end{lemma*}

Essentially, the Kullback-Leibler divergence is of order \(\frac{1}{n}\),
because playing a suboptimal vertex provides no information about which vertex is optimal.
This is unlike the case of the undirected graph, where the optimal vertex is always more likely to be contained in the feedback that the player receives, and the KL divergence does not decay with $n$. The proof of Lemma~\ref{lemma: directed kl} is provided in Appendix~\ref{AppDirectKL}.

By Proposition~\ref{PropLB}, we then have
\begin{align*}
\inf_{\sourceset \in \mathscr{P}} \sup_{\advset \in \mathscr{A}} \ER(\advset, \sourceset) \ge \frac{(n - 1) c d}{n^{3}} \delta T
\left(\frac{n - 1}{n} - \delta \sqrt{\frac{T}{2n}} \sqrt{\frac{c(n - d)}{n(n - c - d)}}\right).
\end{align*}
Finally, we optimize over \(\delta\), \(c\), and \(d\). We take
\[
\delta 
= 
\frac{1}{2} \left(\frac{n - 1}{n}\right) \sqrt{\frac{2n}{T}} \sqrt{\frac{n(n - c - d)}{c (n - d)}},
\]
to obtain
\begin{align*}
\inf_{\sourceset \in \mathscr{P}} \sup_{\advset \in \mathscr{A}} \ER(\advset, \sourceset)
&\geq 
\frac{(n - 1) cd}{4 n^{3}} \left(\frac{n - 1}{n}\right)^{2} T\sqrt{\frac{2n}{T}} \sqrt{\frac{n (n - c - d)}{c (n - d)}} \\ 
&=
\frac{1}{2 \sqrt{2}} \sqrt{n T} \left(\frac{n - 1}{n}\right)^{3} \frac{cd}{n^{2}} \sqrt{\frac{\left(1 - c / n - d / n\right)}{(c / n) \left(1 - d / n\right)}} \\
&\geq 
\frac{1}{16 \sqrt{2}} \sqrt{n T} \frac{cd}{n^{2}} \sqrt{\frac{\left(1 - c / n - d / n\right)}{(c / n) \left(1 - d / n\right)}},
\end{align*}
where the last inequality uses the bound \(\frac{n-1}{n} \ge \frac{1}{2}\). Finally, using the fact that the function
\begin{equation*}
f(x,y) = xy \sqrt{\frac{1-x-y}{x(1-y)}}
\end{equation*}
achieves its maximum value of $\frac{1}{3\sqrt{3}}$ when $(x,y) = \left(\frac{1}{6}, \frac{2}{3}\right)$, we obtain the bound
\[
\inf_{\sourceset \in \mathscr{P}} \sup_{\advset \in \mathscr{A}} \ER(\advset, \sourceset)
\geq \frac{1}{48 \sqrt{6}} \sqrt{T n},
\]
when \(c = \frac{n}{6}\) and \(d = \frac{2n}{3}\).

\subsection{Proofs of KL bounds}
\label{section: kl proofs}

In this Appendix, we derive the required upper bounds on the KL divergence between adversarial strategies. We begin by proving a useful technical lemma.

Recall that $\mprob_i$ denotes the distribution of the edge feedback $\mathscr{I}^T$ under strategy $\advset^i$, and $\sourceset \in \mathscr{P}_d$ is a fixed deterministic player strategy. Also recall that $T_i = |\{t: \sourceset_t = \{i\}\}|$ denotes the number of times vertex $i$ is chosen by the player.

Let \(\prob_i^{t}\) denote the distribution of the edge feedback $\mathscr{I}^t$ under strategy $\advset^i$, so $\mprob_i = \mprob_i^T$. For each pair of nodes $i$ and $v$ and any $1 \le t \le T$, define the function $KL_i^t(v)$ to be the KL divergence between the edge feedback, conditioned on any \(\mathscr{I}^{t - 1}\) such that \(\sourceset_{t} = \{v\}\):
\[
KL_{i}^t(v)
= 
KL
\left(
\prob_0^{t}\{\cdot | \mathscr{I}^{t - 1}\}, \prob_i^{t}\{\cdot | \mathscr{I}^{t - 1}\}
\right).
\]
Note that $KL_i^t(v)$ is indeed a well-defined function of $v$, since conditioned on $\mathscr{I}^{t-1}$, the player's action $\sourceset_t$ is deterministic. Hence, the randomness in $\mathscr{I}^t$ is purely due to the stochastic action of the adversary at time $t$.

\begin{lemma*}
If \(KL_{i}^t(v)\) is independent of \(t\), we have
\begin{equation}
\label{EqnKLfirst}
KL\left(
\prob_0, \prob_i
\right)
= 
KL_{i}(i) \expect_0[T_{i}]
+
\sum_{j \neq i} KL_{i}(j) \expect_0[T_{j}].
\end{equation}
If in addition \(KL_{i}(i)\) is independent of \(i\), for $1 \le i \le n$, and \(KL_{i}(j)\) is constant for all nonzero pairs $i \neq j$, we have
\begin{equation}
\label{EqnKLsecond}
\sum_{i = 1}^{n}
KL\left(\prob_0, \prob_i\right)
=
KL_{i}(i) T
+  KL_{i}(j) (n - 1) T.
\end{equation}
\label{lemma: kl general}
\end{lemma*}

\begin{proof}{Proof.}
Note that equation~\eqref{EqnKLsecond} follows immediately from equation~\eqref{EqnKLfirst} by summing over $i$ and using the fact that $\sum_{i=1}^n \E_0[T_i] = T$.

To derive equation~\eqref{EqnKLfirst}, we use the chain rule for KL divergence:
\begin{align*}
KL\left(\prob_0, \prob_i\right) & =
\sum_{t = 1}^{T} \sum_{\mathscr{I}^{t - 1}} \prob_0 \left\{\mathscr{I}^{t - 1}\right\} 
KL\left(
\prob^{t}_0 \left\{\cdot | \mathscr{I}^{t - 1}\right\}, 
\prob^{t}_i\left\{\cdot | \mathscr{I}^{t - 1}\right\}
\right) \\
& = \sum_{t=1}^T \sum_{v = 1}^n \sum_{\mathscr{I}^{t-1}: \sourceset_t = \{v\} }\mprob_0\{\mathscr{I}^{t-1}\} KL_i^t(v) \\
& \stackrel{(a)}{=} \sum_{t=1}^T \sum_{v = 1}^n \mprob_0\{\sourceset_t = \{v\}\} KL_i(v) \\
& = \sum_{t=1}^T \mprob_0\{\sourceset_t = \{i\}\} KL_i(i) + \sum_{t=1}^T \sum_{j \neq i} \mprob_0\{\sourceset_t = \{j\}\} KL_i(j),
\end{align*}
using the assumption that $KL_i^t(v)$ is independent of $t$ in the equation $(a)$. Now we simply recognize that
\begin{equation*}
\E_0[T_i] = \E_0\left[\sum_{t=1}^T \ind_{\sourceset_t = \{i\}}\right] = \sum_{t=1}^T \mprob_0 \{\sourceset_t = \{i\}\}
\end{equation*}
to obtain the desired equality.
\end{proof}

\subsubsection{Proof of Lemma~\ref{lemma: clique kl computation}}
\label{AppCliqueKL}

Note that $KL_i^t(v)$ is independent of $t$, since the adversary's actions are i.i.d.\ across time steps. Furthermore, $KL_i(i)$ is clearly independent of $i$ and $KL_i(j)$ is constant for all pairs $i \neq j$, so equation~\eqref{EqnKLsecond} of Lemma~\ref{lemma: kl general} holds.

We first compute an upper bound for $KL_i(i)$. Let $X$ denote the size of the connected component containing $i$ on a particular time step, based on the edges played by the adversary. Then
\begin{equation*}
KL_i(i) = KL(\mprob_0(X), \mprob_i(X)),
\end{equation*}
where we abuse notation slightly and write $\mprob_i(X)$ to denote the distribution of $X$ under adversarial strategy $\advset^i$. Also let $Y$ be the indicator variable that $i$ is in the clique selected by the adversary. By the chain rule for the KL divergence,
\begin{equation*}
KL(\mprob_0(X), \mprob_i(X)) \le KL(\mprob_0(X,Y), \mprob_i(X,Y)).
\end{equation*}
We will derive an upper bound for the latter quantity. In particular, the range of $(X,Y)$ is
\begin{equation*}
\{(1,0), (1,1)\} \cup \{(m,1): 2 \le m \le n\}.
\end{equation*}
This leads to the following expression for $KL(\mprob_0(X,Y), \mprob_i(X,Y))$:
\begin{align*}
& \left(1 - \frac{c}{n}(1 - \delta) \right)
\log \left(
\frac{1 - \frac{c}{n}(1 - \delta)}{1 - \frac{c}{n}}
\right) \\
& \qquad + \left(\frac{c}{n}(1-\delta)\left(1-\frac{c}{n}(1-\delta)\right)^{n-1}\right) \log\left(\frac{\frac{c}{n} (1-\delta) \left(1 - \frac{c}{n}(1-\delta)\right)^{n-1}}{\frac{c}{n} \left(1 - \frac{c}{n}(1-\delta)\right)^{n-1}}\right) \\ 
&\qquad 
+ 
\sum_{m = 2}^n
\binom{n - 1}{m-1}
\frac{c}{n} (1 - \delta) 
\left(\frac{c}{n} (1 - \delta) \right)^{m-1}
\left(1 - \frac{c}{n}(1 - \delta) \right)^{n - m}
\\ &\qquad \quad
\times 
\log 
\left(
\frac{\frac{c}{n} (1 - \delta) 
\left(\frac{c}{n} (1 - \delta) \right)^{m-1}
\left(1 - \frac{c}{n}(1 - \delta) \right)^{n - m}
 }{\frac{c}{n} 
\left(\frac{c}{n} (1 - \delta) \right)^{m-1}
\left(1 - \frac{c}{n}(1 - \delta) \right)^{n - m}
}
\right) \\ 
& = \left(1 - \frac{c}{n}(1 - \delta) \right)
\log \left(
\frac{1 - \frac{c}{n}(1 - \delta)}{1 - \frac{c}{n}}
\right) \\
& \qquad + \sum_{m = 1}^n
\binom{n - 1}{m-1}
\frac{c}{n} (1 - \delta) 
\left(\frac{c}{n} (1 - \delta) \right)^{m-1}
\left(1 - \frac{c}{n}(1 - \delta) \right)^{n - m} \log(1-\delta) \\
&=
\left(1 - \frac{c}{n}(1 - \delta) \right) 
\log\left(\frac{1-\frac{c}{n}(1-\delta)}{1-\frac{c}{n}}\right)
+ \frac{c}{n} (1 - \delta) \log \left(1 - \delta \right).
\end{align*}
Applying the inequality \(\log(1 + x) \leq x\) twice, we then obtain
\begin{align}
KL(\mprob_0(X), \mprob_i(X))
&\leq \left(1-\frac{c}{n}(1-\delta)\right) \frac{\frac{c\delta}{n}}{1-\frac{c}{n}} - \frac{c}{n} (1-\delta) \delta \notag \\
&=
\frac{c\delta}{n}
\left(\frac{n - c(1 - \delta)}{n - c} -(1 - \delta) \right)  \nonumber \\ 
&= 
\frac{c\delta^2}{n - c}.
\label{eqn: clique kli}
\end{align}

The computation for $KL_i(j)$ is similar. Let $X$ denote the size of the connected component containing $j$, and let $\cC$ denote the clique chosen by the adversary. Define the random variable
\begin{equation*}
Y = \begin{cases}
0, & \text{if } j \notin \cC \\
1, & \text{if } j \in \cC \text{ and } i \notin \cC \\
2, & \text{if } i,j \in \cC.
\end{cases}
\end{equation*}
Again, it suffices to obtain a bound on $KL(\mprob_0(X,Y), \mprob_i(X,Y))$. The range of $(X,Y)$ is
\begin{equation*}
\{(1,0), (1,1)\} \cup \{(m,1): 2 \le m \le n-1\} \cup \{(m,2): 2 \le m \le n\}.
\end{equation*}
Further note that $\mprob_0(1,0) = \mprob_i(1,0)$, so we may ignore this term when computing the KL divergence.
We then have following expression for $KL(\mprob_0(X,Y), \mprob_i(X,Y))$:
\begin{align*}
&\left(\frac{c}{n}(1-\delta)\right)\left(1 - \frac{c}{n}(1-\delta)\right)^{n-1} \log\left(\frac{\left(\frac{c}{n}(1-\delta)\right)\left(1-\frac{c}{n}(1-\delta)\right)^{n-1}}{\left(\frac{c}{n}(1-\delta)\right)\left(1-\frac{c}{n}\right)\left(1-\frac{c}{n}(1-\delta)\right)^{n-2}}\right) \\
&\qquad +
\sum_{m = 2}^{n-1}
\binom{n - 2}{m-1}
\frac{c}{n}(1 - \delta) \left(1 - \frac{c}{n}(1 - \delta)\right)
\left(\frac{c}{n}(1 - \delta)\right)^{m-1}
\left(1 - \frac{c}{n}(1 - \delta)\right)^{n - m - 1} \\
&\qquad \quad \times 
\log 
\left(
\frac{\frac{c}{n}(1 - \delta) \left(1 - \frac{c}{n}(1 - \delta)\right)
\left(\frac{c}{n}(1 - \delta)\right)^{m-1}
\left(1 - \frac{c}{n}(1 - \delta)\right)^{n - m - 1}     
}{\frac{c}{n}(1 - \delta) \left(1 - \frac{c}{n} \right)
\left(\frac{c}{n}(1 - \delta)\right)^{m-1}
\left(1 - \frac{c}{n}(1 - \delta)\right)^{n - m - 1}}\right) \\
& \qquad + \sum_{m = 2}^n
\binom{n - 2}{m-2}
\left(\frac{c}{n}(1 - \delta)\right)^{2}
\left(\frac{c}{n}(1 - \delta)\right)^{m-2}
\left(1 - \frac{c}{n}(1 - \delta)\right)^{n - m} \\
&\qquad \quad \times 
\log \left(\frac{\left(\frac{c}{n}(1 - \delta)\right)^{2}
\left(\frac{c}{n}(1 - \delta)\right)^{m-2}
\left(1 - \frac{c}{n}(1 - \delta)\right)^{n - m}
}{\left(\frac{c}{n}(1 - \delta)\right) 
\left(\frac{c}{n} \right)
\left(\frac{c}{n}(1 - \delta)\right)^{m-2}
\left(1 - \frac{c}{n}(1 - \delta)\right)^{n - m}} \right) \\
& = \sum_{m = 1}^{n-1}
\binom{n - 2}{m-1}
\frac{c}{n}(1 - \delta) \left(1 - \frac{c}{n}(1 - \delta)\right)
\left(\frac{c}{n}(1 - \delta)\right)^{m-1}
\left(1 - \frac{c}{n}(1 - \delta)\right)^{n - m - 1} \\
&\qquad \quad \times 
\log 
\left(
\frac{\frac{c}{n}(1 - \delta) \left(1 - \frac{c}{n}(1 - \delta)\right)
\left(\frac{c}{n}(1 - \delta)\right)^{m-1}
\left(1 - \frac{c}{n}(1 - \delta)\right)^{n - m - 1}     
}{\frac{c}{n}(1 - \delta) \left(1 - \frac{c}{n} \right)
\left(\frac{c}{n}(1 - \delta)\right)^{m-1}
\left(1 - \frac{c}{n}(1 - \delta)\right)^{n - m - 1}}\right) \\
& \qquad + \sum_{m = 2}^n
\binom{n - 2}{m-2}
\left(\frac{c}{n}(1 - \delta)\right)^{2}
\left(\frac{c}{n}(1 - \delta)\right)^{m-2}
\left(1 - \frac{c}{n}(1 - \delta)\right)^{n - m} \\
&\qquad \quad \times 
\log \left(\frac{\left(\frac{c}{n}(1 - \delta)\right)^{2}
\left(\frac{c}{n}(1 - \delta)\right)^{m-2}
\left(1 - \frac{c}{n}(1 - \delta)\right)^{n - m}
}{\left(\frac{c}{n}(1 - \delta)\right) 
\left(\frac{c}{n} \right)
\left(\frac{c}{n}(1 - \delta)\right)^{m-2}
\left(1 - \frac{c}{n}(1 - \delta)\right)^{n - m}} \right) \\
& = \frac{c}{n} (1 - \delta) \left(1 - \frac{c}{n}(1 - \delta) \right)
\log \left(\frac{1 - \frac{c}{n}(1 - \delta)}{1 - \frac{c}{n} }\right) + \left(\frac{c}{n} (1 - \delta) \right)^{2} \log \left(1 - \delta \right).
\end{align*}
We once again use the inequality \(\log(1 + x) \leq x\) to obtain
\begin{align}
KL(\mprob_0(X), \mprob_i(X)) 
&\leq \frac{c}{n}(1-\delta)\left(1-\frac{c}{n}(1-\delta)\right)\left(\frac{\frac{c\delta}{n}}{1-\frac{c}{n}}\right) - \left(\frac{c}{n}(1-\delta)\right)^2\delta \notag \\
&=
\left(\frac{c}{n}\right)^{2} (1 - \delta)  \delta
\left(
\frac{n - c(1 - \delta)}{n - c} - (1 - \delta)
\right) 
\nonumber \\ 
&=
\left(\frac{c}{n}\right)^{2} (1 - \delta)  \delta^{2} \frac{n}{n - c}.
\label{eqn: clique klj}
\end{align}
Combining inequalities~\eqref{eqn: clique kli} and~\eqref{eqn: clique klj} with equation~\eqref{EqnKLsecond} of Lemma~\ref{lemma: kl general}, we obtain the bound
\begin{align*}
\sum_{i = 1}^{n} KL\left(\prob_0, \prob_i\right)
&\leq 
\frac{c}{n - c} \delta^{2} T + \left(\frac{c}{n}\right)^{2} (1 - \delta) \delta^{2} \frac{n}{n - c} (n - 1)  T \\
& \leq \frac{c (c + 1)}{n - c} T\delta^{2},
\end{align*}
completing the proof.

%
\subsubsection{Proof of Lemma~\ref{lemma: directed kl}}
\label{AppDirectKL}

Note that $KL_i^t(v)$ is independent of $t$, since the adversary's actions are i.i.d.\ across time steps. Furthermore, $KL_i(i)$ is clearly independent of $i$ and $KL_i(j)$ is constant for all pairs $i \neq j$, so equation~\eqref{EqnKLsecond} of Lemma~\ref{lemma: kl general} holds.

Note that when \(\sourceset_t = \{j\}\), the distribution of the feedback \(\mathscr{I}_{t}\) is the same under \(\prob^{t}_0\{\cdot | \mathscr{I}^{t - 1}\}\) and \(\prob^{t}_i\{\cdot | \mathscr{I}^{t - 1}\}\), since the vertex $i$ is chosen to be a sink vertex with the same probability $\frac{d}{n}$ under both $\advset^0$ and $\advset^i$. Hence, \(KL_{i}(j) = 0\).

To compute $KL_i(i)$, let $X$ denote the size of the influenced component containing $i$ when $\sourceset_t = \{i\}$, and define the random variable
\begin{equation*}
Y = \begin{cases}
0, & \text{if } i \text{ is a sink vertex} \\
1, & \text{if } i \text{ is a source vertex} \\ 
2, & \text{otherwise}.
\end{cases}
\end{equation*}
As in the proof of Lemma~\ref{lemma: clique kl computation}, we will upper-bound $KL(\mprob_0(X,Y), \mprob_i(X,Y))$, leading to an upper bound on $KL(\mprob_0(X), \mprob_i(X))$. The range of $(X,Y)$ is
\begin{equation*}
\{(1,0), (1,2)\} \cup \{(m,1): 2 \le m \le n\}.
\end{equation*}
We then have the following expression for $KL(\mprob_0(X,Y), \mprob_i(X,Y))$:
\begin{align*}
& \frac{d}{n} \log \left(\frac{d / n}{d / n}\right)
+ \left(1- \frac{c}{n}(1-\delta) - \frac{d}{n}\right) \log \left(\frac{1 - \frac{c}{n}(1-\delta) - \frac{d}{n}}{1 - \frac{c}{n} - \frac{d}{n}}\right)
\\ &\qquad +
\sum_{m = 2}^n \binom{n - 1}{m-1} \frac{c}{n}(1-\delta) \left(\frac{d}{n}\right)^{m-1}
\left(1 - \frac{d}{n}\right)^{n - m} \\ 
&\qquad \quad \times 
\log \left(
\frac{\frac{c}{n}(1-\delta)\left(\frac{d}{n}\right)^{m-1} \left(1-\frac{d}{n}\right)^{n-m}}{\frac{c}{n} \left(\frac{d}{n}\right)^{m-1} \left(1-\frac{d}{n}\right)^{n-m}}
\right), \\ 
&=
\frac{n - c - d + c \delta}{n} \log \left(\frac{n - c - d + c \delta}{n - c - d}\right) \\
& \qquad + \frac{c}{n}(1-\delta) \log(1-\delta) \sum_{m=2}^n \binom{n-1}{m-1} \left(\frac{d}{n}\right)^{m-1} \left(1-\frac{d}{n}\right)^{n-m} \\ 
&\le
\frac{n - c - d + c \delta}{n} \log \left(\frac{n - c - d + c \delta}{n - c - d}\right)
+ \frac{c}{n} (1 - \delta) \log(1 - \delta).
\end{align*}
Using the inequality \(\log(1 + x) \leq x\), we then have
\begin{align*}
KL(\mprob_0(X), \mprob_i(X)) \leq 
\frac{c\delta(n - c - d + c \delta)}{n(n - c - d)} - \frac{c}{n} (1 - \delta) \delta =
\frac{c(n - d)}{n (n - c - d)} \delta^{2}.
\end{align*}
Applying Lemma~\ref{lemma: kl general} completes the proof.

\end{appendix}

\end{document}